\documentclass[11pt]{article}
\usepackage{authblk}
\usepackage{caption}
\usepackage[caption=false]{subfig}

\usepackage{amsmath,amsfonts,amsthm,amssymb,graphics,graphicx,mathtools,epstopdf,etoolbox,indentfirst,enumitem,mleftright,xcolor,booktabs,footnote,relsize,microtype}

\usepackage[scr=euler]{mathalpha} 
\usepackage{accents,yhmath} 

\makeatletter
\renewcommand{\paragraph}{%
  \@startsection{paragraph}{4}%
  {\z@}{1.25ex \@plus 1ex \@minus .2ex}{-1em}%
  {\normalfont\normalsize\bfseries}%
}
\makeatother

\usepackage[
backend=bibtex,
style=alphabetic,
giveninits=true, 
maxbibnames=99,
minalphanames=3,
date=year
]{biblatex}
\usepackage{hyperref} 
\usepackage{doi} 

\renewbibmacro{in:}{}

\AtEveryBibitem{\clearfield{issue}} 

\newbibmacro{string+doiurl}[1]{%
	\iffieldundef{doi}{%
		\iffieldundef{url}{
			#1%
		}{%
			\href{\thefield{url}}{#1}%
		}%
	}{%
		\href{http://doi.org/\thefield{doi}}{#1}%
	}%
}
\DeclareFieldFormat{title}{\usebibmacro{string+doiurl}{\mkbibemph{#1}}}
\DeclareFieldFormat[article,thesis,incollection,inproceedings,manual]{title}%
{\usebibmacro{string+doiurl}
	{#1} 
}
\usepackage{algorithm} 
\usepackage[noend]{algpseudocode}
\floatname{algorithm}{Protocol} 
\algrenewcommand\alglinenumber[1]{\normalsize #1.} 

\newcounter{algsubstate}




\definecolor{darkmagenta}{rgb}{0.85, 0, 0.45}
\hypersetup{
colorlinks = true,
citecolor= blue,
urlcolor= blue,
linkcolor = blue
}

\newcommand{\ket}[1]{\left| #1 \right>}
\newcommand{\bra}[1]{\left< #1 \right|}
\newcommand{\ketbra}[2]{\ket{#1} \!\! \bra{#2}}
\newcommand{\pure}[1]{\ketbra{#1}{#1}}
\newcommand{\tr}[2][]{\operatorname{Tr}_{#1}\!\left[#2\right]} 
\newcommand{\Tr}{\operatorname{Tr}} 

\newcommand{\defvar}{\coloneqq} 
\newcommand{\dop}[1]{\operatorname{S}_{#1}} 
\newcommand{\eps}{\epsilon}
\newcommand{\freq}{\operatorname{freq}}
\newcommand{\id}{\mathbb{I}} 
\newcommand{\idmap}{\operatorname{id}} 

\newcommand{\norm}[1]{\left\lVert#1\right\rVert} 
\newcommand{\pd}{P} 

\newcommand{\Pos}{\operatorname{Pos}} 
\newcommand{\suchthat}{\text{ s.t.}} 
\newcommand{\supp}{\operatorname{supp}} 
\newcommand{\term}[1]{\textup{\textbf{#1}}}

\newcommand{\Renyi}{R\'{e}nyi}
\newcommand{\mbf}[1]{\mathbf{#1}} 
\newcommand{\bsym}[1]{\boldsymbol{#1}} 
\newcommand{\Dmax}{D_\infty}
\newcommand{\Hminup}[1][]{H_\infty^{\uparrow #1}}
\newcommand{\Imaxnone}[1][]{I_\infty^{#1}}
\newcommand{\ImaxDA}[1][]{I_\infty^{\downarrow #1}}
\newcommand{\ImaxDDA}[1][]{I_\infty^{\downarrow\downarrow #1}}
\newcommand{\IDA}{I^\downarrow}
\newcommand{\IDDA}{I^{\downarrow\downarrow}}
\newcommand{\Idiff}{I^{\downarrow,\mathrm{diff}}}

\newcommand{\EATchann}{\mathcal{M}}
\newcommand{\LEAKchann}{\mathcal{L}}
\newcommand{\CS}{\overline{C}} 
\newcommand{\copyCS}{\overline{\overline{C}}}
\newcommand{\cS}{\bar{c}} 
\newcommand{\alphCS}{\overline{\mathcal{C}}}
\newcommand{\CP}{\widehat{C}} 

\newcommand{\cP}{\hat{c}} 
\newcommand{\alphCP}{\widehat{\mathcal{C}}}

\newcommand{\stabR}{\widetilde{Z}}

\newtheorem{remark}{Remark}[section]
\newtheorem{theorem}{Theorem}[section]
\newtheorem{lemma}{Lemma}[section]
\newtheorem{corollary}{Corollary}[section]

\newtheorem{fact}{Fact}[section]
\theoremstyle{definition} 
\newtheorem{definition}{Definition}[section]

\newcommand\numberthis{\addtocounter{equation}{1}\tag{\theequation}}

\begin{document}

\title{\textbf{Mutual information chain rules for security proofs robust against device imperfections}}
\renewcommand\Affilfont{\itshape\small} 

\author[1]{Amir Arqand}
\author[2]{Tony Metger}
\author[1]{Ernest Y.-Z.\ Tan}
\affil[1]{Institute for Quantum Computing and Department
of Physics and Astronomy, University of Waterloo, Waterloo, Ontario N2L 3G1, Canada.}
\affil[2]{Institute for Theoretical Physics, ETH Z{\"u}rich, 8093 Z{\"u}rich, Switzerland}

\date{}

\maketitle

\begin{abstract}
In this work we derive a number of chain rules for mutual information quantities, suitable for analyzing quantum cryptography with imperfect devices that leak additional information to an adversary. First, we derive a chain rule between smooth min-entropy and smooth max-information, which improves over previous chain rules for characterizing one-shot information leakage caused by an additional conditioning register. Second, we derive an ``information bounding theorem'' that bounds the Rényi mutual information of a state produced by a sequence of channels, in terms of the Rényi mutual information of the individual channel outputs, similar to entropy accumulation theorems. In particular, this yields simple bounds on the smooth max-information in the preceding chain rule. Third, we derive chain rules between Rényi entropies and Rényi mutual information, which can be used to modify the entropy accumulation theorem to accommodate leakage registers sent to the adversary in each round of a protocol. We show that these results can be used to handle some device imperfections in a variety of device-dependent and device-independent protocols, such as randomness generation and quantum key distribution.
\end{abstract}

\section{Introduction}
Standard security proofs for quantum cryptography protocols (such as quantum key distribution (QKD) or randomness generation) typically assume that after a measurement is performed, the resulting raw key data remains secret and independent of various processes in subsequent rounds (apart from standard techniques to handle publicly announced classical information)~\cite{rennerthesis}. However, when moving towards practical device implementations, such an assumption may no longer be a proper description of the actual device,  
due to issues such as photon leakage from photodetectors~\cite{KZMW01,PCS+18}, or source imperfections and correlations~\cite{LPK23,MCA23,MD24,CNLT24} (in which the states generated by a source in a prepare-and-measure protocol are not exactly the desired ones, or have correlations to previous setting choices).
It would hence be more realistic to take into account possible imperfections of the device in security proofs. 
Thus, having a method to handle security proofs with imperfect devices is crucial. Such a method should also be robust against a small amount of imperfection. While such a property might sound intuitive, we highlight that phenomena such as information locking~\cite{KRBM07,DHL+04,Win17} demonstrate that there are situations where even a small amount of leaked information can have drastic effects. 

To handle this issue, in this work we propose a variety of chain rules involving mutual information quantities, which can quantify the amount of information leakage from imperfect devices. While this is most intuitively applied to photon leakage or similar mechanisms, we show that it can also be applied to other imperfections such as source correlations as well. From an information-theoretic perspective, our bounds are ``natural'' in the sense that they involve mutual-information quantities, quantifying the amount of correlation between the leakage and the secret data.
They should also be simple to apply in practice, as our final results simply consist of directly subtracting off some easily computed corrections from the keyrates obtained from existing entropic proof techniques. 

We start by proving a chain rule that connects smooth min-entropy and smooth max-information, to characterize the information leakage caused by an extra conditioning register in a one-shot setting. (The smooth min-entropy is the quantity that characterizes extractable secure key length under a composable security definition, hence avoiding issues such as information locking.) Previously, one approach to remove a conditioning leakage register in smooth min-entropy was to use a fairly crude chain rule that just subtracts off the log-dimension of the leakage register. In~\cite{VDT13}, a tighter bound was proved, in which the leakage information was measured by smooth max-entropy. However, even though that bound is tighter than just the log-dimension, it is still not the tightest bound. Intuitively, when leakage happens, what should matter is the amount of its correlation to the secret data. 
This observation is supported by the chain rules for von Neumann entropies, where the effect of removing a conditioning register is characterized by conditional quantum mutual information (CQMI)~\cite{CA98}. 

Unfortunately, there has not yet been a unique well-behaved definition of smooth max-CQMI, and in this paper we do not aim to introduce a new definition for this quantity. However, it is well known that an upper bound on CQMI in the von Neumann regime is the mutual information between appropriately chosen registers, which is also a measure of correlation~\cite{KW20,DH11,DHW13}; furthermore, unlike smooth max-CQMI, there are well-behaved definitions of smooth max-information~\cite{CBR14}. Motivated by these results, we show that one can at least bound the effect of leakage via a smooth max-information between the leakage register and all the other registers; specifically, we prove a chain rule of the form
\begin{align}\label{eq:1shotinformal}
\text{(informal summary)} \quad
H_\mathrm{min}^{\epsilon}(S|LE)_\rho\geq H_\mathrm{min}^{\epsilon'}(S|E)_\rho-I_\mathrm{max}^{\epsilon''}(SE;L)_\rho- \text{[small corrections]},
\end{align}
where the registers can be interpreted as $S$ being the secret data, $L$ being the leakage registers, and $E$ being some other side-information. We furthermore argue that in fact the leakage term can be sharpened to something roughly like a smooth max-CQMI, though the resulting quantity may be difficult to analyze and we do not attempt to do so here. 

The smooth max-information in the above context would usually involve a large number of registers leaked over the course of a protocol. To obtain a simple bound on it, we next derive an ``information bounding theorem''. Informally, suppose we have a bipartite state that is produced by applying a sequence of channels to some initial state, and we are interested in bounding the \Renyi\ mutual information between the registers in the final state. Then, our information bounding theorem provides us with a simple bound based on \Renyi\ mutual information of the individual rounds. Specifically, we derive a bound of the following form:
\begin{align}\label{eq:introRenyiIAT}
	\IDA_\alpha(Z;L_1^n)_\rho\leq\sum_{i=1}^n\sup_{\omega_{\stabR R_{i-1}}}\IDA_\alpha(\stabR ;L_i)_{\mathcal{M}_i(\omega)},
\end{align}
where $\{\mathcal{M}_i:R_{i-1}\rightarrow R_iL_i\}_i$ are a sequence of quantum channels such that $\rho=\mathcal{M}_n\circ\cdots\circ\mathcal{M}_1(\rho^0)$ for some initial state $\rho^0$, and $Z$ is some overall ``reference'' register not acted on by the channels, while $\stabR$ is a purifying register for any individual channel.
With this in hand, and existing results that connect smooth max-information to the \Renyi\ mutual information, we can bound the smooth max-information via a bound of the form
\begin{align}\label{eq:introIAT}
\text{(informal summary)} \quad  I_\mathrm{max}^{\epsilon}(Z;L_1^n)_\rho \leq\sum_{i=1}^n\sup_{\omega_{\stabR R_{i-1}}}I(\stabR ;L_i)_{\mathcal{M}_i(\omega)} + O(\sqrt{n}),
\end{align}
connecting the smooth max-information to the von Neumann mutual information. This can then be applied in the chain rule~\eqref{eq:1shotinformal} for some protocols, under some assumptions about their structure.

This result is similar to a form of the entropy accumulation theorem (EAT)~\cite{DFR20} or its generalized version (GEAT)~\cite{MFSR22}, except that here the quantity of interest is mutual information, instead of conditional entropy or divergence.\footnote{We note that the information bounding theorem does not seem to be a straightforward consequence of the {\Renyi} divergence chain rules in the GEAT~\cite{MFSR22} (although our proof techniques are instead somewhat similar to those for the original EAT~\cite{DFR20}). This is because the GEAT chain rules are formulated in terms of channel divergences, i.e.~optimizations over all possible input states to a channel, and it does not seem straightforward to construct a channel that outputs states with the correct ``product structure'' to yield a {\Renyi} mutual information in the sense of $D_\alpha(\rho_{AB}||\rho_A\otimes\rho_B)$ or $\inf_{\sigma_B} D_\alpha(\rho_{AB}||\rho_A\otimes\sigma_B)$ or $\inf_{\omega_A,\sigma_B} D_\alpha(\rho_{AB}||\omega_A\otimes\sigma_B)$ (see Def.~\ref{def:sandwiched entropy and info}).} (Also, since here we are \emph{upper}-bounding the final quantity of interest, we refer to it as a ``bounding'' theorem rather than ``accumulation''; see also Remark~\ref{remark:MMI}.) It should be noted that this information bounding theorem is a stand-alone statement that can bound the \Renyi\ mutual information in an $n$-round process, and thus might have other applications in information theory, in scenarios where one might be interested in breaking down the $n$-round \Renyi\ mutual information into its individual rounds.

The preceding results have some limitations on the way one can handle leakage processes with them, in that they essentially analyze the leakage registers as a ``separate'' process from the main protocol. This is restrictive, since for some device imperfections a more realistic description of the protocol would be to allow the leakage to happen during the protocol itself, and it does not seem straightforward to apply the above tools in various such cases.\footnote{However, for some tasks such as randomness generation or expansion, where one sometimes assumes that an adversary does not interact with the devices apart from holding some initial side-information register, the preceding results would indeed already suffice.} For this purpose, we derive yet another chain rule that relates \Renyi\ conditional entropies to \Renyi\ mutual information, in the sense that for states satisfying a suitable channel structure, we can write
\begin{align}
\text{(informal summary)} \quad 
H_\alpha^\uparrow(A|BC)_\rho \geq H_\alpha^\uparrow(A|B)_\rho-\xi_\alpha,
\end{align}
where $\xi_\alpha$ is essentially a {\Renyi} mutual information involving $C$, after taking the ``worst-case'' value over possible channel outputs.
We apply this chain rule iteratively over sequences of channels to develop variations of entropy accumulation, where leakage is allowed to happen during the protocol and an adversary can interact with it adaptively, and we compensate for it by simply subtracting the sum of $\xi_\alpha$ values across rounds. 

\begin{figure}
	\centering
	\includegraphics[width =\columnwidth]{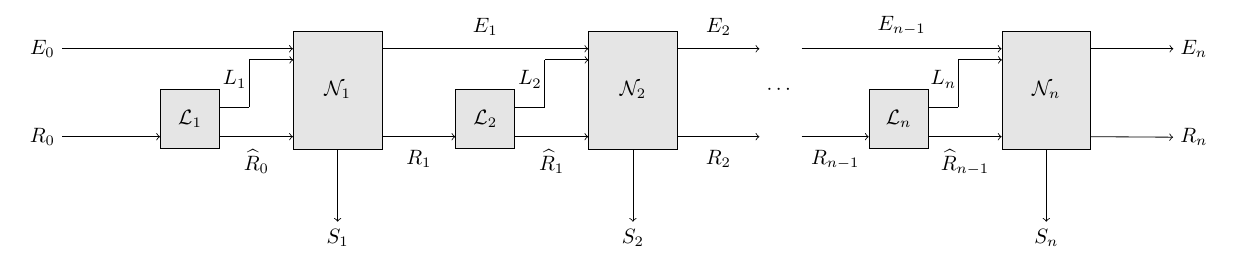}
	\caption{Schematic setup of a QKD protocol with leakage. The registers $E_i$ are the adversary's side information, $L_i$ represent some ``leaked'' information registers in each round, and $R_i$ are memory registers stored between rounds without being available to the adversary except via the leakage. The channels  $\mathcal{L}_i$ describe the leakage processes, while the channels $\mathcal{N}_i$ represent all other actions in each round that satisfy a non-signalling condition from the memory registers to the side-information registers. (Note that this setup is most compatible with the framework of entropy accumulation with leakage which we describe in Sec.~\ref{subsec:chainwithleakage} and Appendix~\ref{app:great}.)}
	\label{fig.leakageEAT}
\end{figure}

This result allows us to analyze protocols where a leakage process happens in each round, and the adversary is allowed to adapt their attacks arbitrarily using the leaked information (Fig.~\ref{fig.leakageEAT}). Furthermore, our results should be very simple to apply in practice, since they consist entirely of just subtracting some easily-computed value from the existing keyrate calculations. We describe some simple examples of device imperfections we can handle with our approach, though we leave a detailed analysis for future work. As our analysis is compatible with the entropy accumulation framework, it applies to both device-dependent and device-independent (DI) protocols, with devices that have imperfections outside of the ``standard assumptions'' in those scenarios.

\subsection{Prior work}

A {\Renyi} chain rule with a somewhat similar form to our one-shot chain rule in Eq.~\eqref{eq:1shotinformal} was derived in~\cite{MT20}. Specifically, from their results, we can write the following (or some variants depending on which {\Renyi} conditional entropy or mutual information definitions are used) for suitable {\Renyi} parameters $\alpha,\alpha',\alpha''$:
\begin{align}
H_{\alpha''}(S|LE)_\rho \geq H_{\alpha'}(S)_\rho - \IDA_\alpha(LE;S)_\rho \geq H_{\alpha'}(S|E)_\rho - \IDA_\alpha(LE;S)_\rho ,
\end{align} 
where the second inequality holds by data-processing. Alternatively, by applying another instance of their chain rule for some other $\alpha''',\alpha''''$, we could instead get
\begin{align}
H_{\alpha''}(S|LE)_\rho \geq H_{\alpha'}(S)_\rho - \IDA_\alpha(LE;S)_\rho \geq H_{\alpha'''}(S|E)_\rho + \IDA_{\alpha''''}(E;S)_\rho - \IDA_\alpha(LE;S)_\rho .
\end{align} 
However, inspecting the grouping of registers in the above ``leakage'' terms shows that they differ nontrivially from our chain rule~\eqref{eq:1shotinformal}, in a way that is harder to use in our intended applications. 

Furthermore, 
the question of quantum cryptography with leaky devices has been considered in various works:~\cite{KZMW01,PCS+18,TZWP22,TPWP21,PPWT22} considered various forms of leakage behaviour in IID scenarios,~\cite{JK22} considered a leakage model accounting for a bounded number of qubits, and~\cite{arx_Tan23} considered various forms of ``probabilistic'' constrained leakage. Our results should cover a more general class of models than those works.
Also, various source imperfection models were studied in~\cite{LPK23,MCA23,CNLT24}, while approximate versions of entropy accumulation (which could handle various device imperfections) were developed in~\cite{MD23}, and a related approach was used in~\cite{MD24} to analyze some forms of source imperfection. The class of possible imperfections we consider is not straightforwardly comparable to theirs: for instance,~\cite{LPK23,MCA23,CNLT24} consider the closeness of the generated states to some target states, ~\cite{MD23} uses a diamond-norm constraint together with a dimension bound, while~\cite{MD24} uses a form of ``testing'' on the source. In contrast, we choose to quantify the leakage by bounding the {\Renyi} mutual information. The constraints in our approach could perhaps be obtained from theirs or vice versa in some scenarios, but the conversions would be complicated and likely to be highly suboptimal (with the exception of the imperfections considered in~\cite{LPK23,MCA23,CNLT24}, which we believe could also be incorporated into our approach in principle; see Sec.~\ref{sec:examples} later). However, we believe that our results should at least yield simpler bounds in the contexts where they apply.

\subsection{Organization}

The rest of the paper is organized as follows. In Sec.~\ref{sec:notation} we lay out the notations and definitions. In Sec.~\ref{sec:Imax_chain} we present the chain rule between smooth min-entropy and smooth max-information, then follow up in  Sec.~\ref{sec:Info_Acc} with the information bounding theorem that allows us to bound the smooth max-information. In Sec.~\ref{sec:GEATleakage} we present the version of entropy accumulation that accommodates leakage of information within each round. Finally, in Sec.~\ref{sec:examples} we describe some examples of device imperfections that can be handled with our techniques.

\section{Preliminaries}
\label{sec:notation}

\begin{table}
\caption{List of notation}\label{tab:notation}
\def\arraystretch{1.5} 
\setlength\tabcolsep{.28cm}
\begin{tabular}{c l}
\toprule
\textit{Symbol} & \textit{Definition} \\
\toprule
$\log$ & Base-$2$ logarithm \\
\hline
$H$ & Base-$2$ von Neumann entropy \\
\hline
$\norm{\cdot}_p$ & Schatten $p$-norm \\
\hline
$\left|\cdot\right|$ & Absolute value of operator; $\left|M\right| \defvar \sqrt{M^\dagger M}$ \\
\hline
$A\perp B$ & $A$ and $B$ are orthogonal; $AB=BA=0$ \\
\hline
$X\geq Y$ (resp.~$X>Y$) & $X-Y$ is positive semidefinite (resp.~positive definite)\\
\hline
$\Pos(A)$
& Set of positive semidefinite operators on register $A$\\
\hline
$\dop{=}(A)$ (resp.~$\dop{\leq}(A)$) & Set of normalized (resp.~subnormalized) states on register $A$ \\
\hline
$A_j^k$ & Registers $A_j \dots A_k$ \\
\toprule
\end{tabular}
\def\arraystretch{1}
\end{table}

We list some basic notation in Table~\ref{tab:notation}.
Apart from the notation in 
that table,
we will also need to use some other concepts, which we shall define below, and briefly elaborate on in some cases. 
In this work, we will assume that all systems are finite-dimensional, but we will not impose any bounds on the system dimensions unless otherwise specified. 
All entropies are defined in base~$2$.
For a channel $\mathcal{E}$ from register $Q$ to register $Q'$, we will often write the abbreviated notation
\begin{align}
\mathcal{E}: Q \to Q',
\end{align}
rather than writing out the formal statement of it being a CPTP linear map $\mathcal{E}: \operatorname{End}(\mathcal{H}_Q) \to \operatorname{End}(\mathcal{H}_{Q'})$ (where $\operatorname{End}(\mathcal{H}_Q)$ is the set of linear operators on the Hilbert space $\mathcal{H}_{Q}$).
Also, throughout this work we will often leave tensor products with identity channels implicit; e.g.~given a channel $\mathcal{E}:Q\to Q'$, we often use the compact notation
\begin{align}
\mathcal{E}(\rho_{QR}) \defvar (\mathcal{E} \otimes \idmap_R)(\rho_{QR}).
\end{align}

\begin{definition}
A state $\rho \in \dop{\leq}(CQ)$ is said to be \term{classical on $C$} (with respect to a specified basis on $C$) 
if it is in the form 
\begin{align}
\rho_{CQ} = \sum_c \lambda_c \pure{c} \otimes \sigma_c,
\label{eq:cq}
\end{align}
for some normalized states $\sigma_c \in \dop{=}(Q) $ and weights $\lambda_c \geq 0$, with $\ket{c}$ being the 
specified basis states on $C$. In most circumstances, we will not explicitly specify this ``classical basis'' of $C$, 
leaving it to be implicitly defined by context.
It may be convenient to absorb the weights $\lambda_c$ into the states $\sigma_c$, writing them as subnormalized 
states $\omega_c = \lambda_c\sigma_c \in \dop{\leq}(Q)$ instead. 
\end{definition}

\begin{definition}\label{def:cond}
(Conditioning on classical events) For a state $\rho \in \dop{\leq}(CQ)$ classical on $C$, written in the form
$\rho_{CQ} = \sum_c \pure{c} \otimes \omega_c$ 
for some $\omega_c \in \dop{\leq}(Q)$,
and an event $\Omega$ defined on the register $C$, we will define a corresponding \term{partial state} and \term{conditional state} as, respectively,
\begin{align}
\rho_{\land\Omega} \defvar \sum_{c\in\Omega} \pure{c} \otimes \omega_c, \qquad\qquad \rho_{|\Omega} \defvar \frac{\tr{\rho}}{\tr{\rho_{\land\Omega}}} \rho_{\land\Omega} = \frac{
\sum_{c} \tr{\omega_c}
}{\sum_{c\in\Omega} \tr{\omega_c}} \rho_{\land\Omega} .
\end{align}
The process of taking partial states is commutative and ``associative'', in the sense that for any events $\Omega,\Omega'$ we have $(\rho_{\land\Omega})_{\land\Omega'} = (\rho_{\land\Omega'})_{\land\Omega} = \rho_{\land(\Omega\land\Omega')}$; hence for brevity we will denote all of these expressions as
\begin{align}
\rho_{\land\Omega\land\Omega'} \defvar (\rho_{\land\Omega})_{\land\Omega'} = (\rho_{\land\Omega'})_{\land\Omega} = \rho_{\land(\Omega\land\Omega')}.
\end{align}
On the other hand, some disambiguating parentheses are needed when combined with taking conditional states (due to the normalization factors).
\end{definition}

In light of the preceding two definitions, for a normalized state $\rho \in \dop{=}(CQ)$ that is classical on $C$, it is reasonable to write it in the form
\begin{align}
\rho_{CQ} = \sum_c \rho(c) \pure{c} \otimes \rho_{Q|c},
\label{eq:cqstateprobs}
\end{align}
where $\rho(c)$ denotes the probability of $C=c$ according to $\rho$, and $\rho_{Q|c}$ can indeed be interpreted as the conditional state on $Q$ corresponding to $C=c$, i.e.~$\rho_{Q|c} = \tr[C]{\rho_{|\Omega}}$ where $\Omega$ is the event $C=c$. 

\begin{definition}
	\label{def:cond_state}
	For any $\rho\in\dop{=}(AB)$ we write $\rho_{A|B}$ to denote \term{the state of $A$ conditioned on $B$}, which is defined as:
	\begin{align}
		\rho_{A|B}\defvar\rho_B^{-\frac{1}{2}}\rho_{AB}\rho_B^{-\frac{1}{2}}.
	\end{align} 
\end{definition}

We highlight that there is a minor terminology issue regarding ``conditional states'' here, in that Def.~\ref{def:cond} involves states conditioned on an \emph{event} (or a specific value of some classical register), while Def.~\ref{def:cond_state} involves states conditioned on a \emph{register}. Qualitatively, to draw an analogue to classical probability distributions $P_{AB}$ on random variables $AB$, the former is for instance about the notion of the conditional distributions $P_{A|b}$ conditioned on a specific event $B=b$, while the latter is about the notion of the conditional distribution $P_{A|B}$, which is a random variable that is a function of the random variable $B$. In most parts of this work we will only be needing the latter (Def.~\ref{def:cond_state}), and in places such as Fact~\ref{fact:classmix} and Appendix~\ref{app:great} where we need the former (Def.~\ref{def:cond}), the intended interpretation should be clear from context.

The following definitions of {\Renyi} divergences and entropies are reproduced from~\cite{Tom16}, and coincide with those in~\cite{DFR20,DF19,arx_MFSR22} for normalized states.

\begin{definition}\label{def:sandwiched divergence}
({\Renyi} divergence)
For any $\rho,\sigma\in\Pos(A)$ with $\tr{\rho}\neq0$, and $\alpha\in(0,1)\cup (1,\infty)$, the (sandwiched) \Renyi\ divergence between $\rho$, $\sigma$ is defined as:
\begin{align}
    \label{eq:sand_renyi_div}
    D_\alpha(\rho||\sigma)\defvar\begin{cases}
    \frac{1}{\alpha-1}\log\frac{\tr{ \left(\sigma^{\frac{1-\alpha}{2\alpha}}\rho\sigma^{\frac{1-\alpha}{2\alpha}}\right)^\alpha}}{\tr{\rho}} &\left(\alpha < 1\ \wedge\ \rho\not\perp\sigma\right)\vee \left(\supp(\rho)\subseteq\supp(\sigma)\right) \\ 
    +\infty & \text{otherwise},
    \end{cases}  
\end{align}
where for $\alpha>1$ the $\sigma^{\frac{1-\alpha}{2\alpha}}$ terms are defined via the Moore-Penrose pseudoinverse if $\sigma$ is not full-support~\cite{Tom16}.
The above definition is extended to $\alpha \in \{0,1,\infty\}$ by taking the respective limits, and 
the $\alpha=\infty$ case is usually referred to as the \term{max-divergence}.
For the $\alpha=1$ case, it reduces to the Umegaki divergence:
\begin{align}
    \label{eq:umegaki_div}
    D(\rho||\sigma)\defvar\begin{cases}
        \frac{\tr{\rho\log\rho-\rho\log\sigma}}{\tr{\rho}} & \supp(\rho)\subseteq\supp(\sigma)\\
    +\infty & \text{otherwise}. 
    \end{cases}
\end{align}
For any two classical probability distributions $\mbf{p},\mbf{q}$ on a common alphabet, the {\Renyi} divergence $D_\alpha(\mbf{p}||\mbf{q})$ is defined analogously, e.g.~by viewing the distributions as diagonal density matrices in the above formulas; in the $\alpha=1$ case this gives the Kullback–Leibler (KL) divergence.
\end{definition}
\begin{definition}\label{def:stablized channel divergence} 
(Stabilized {\Renyi} channel divergence)
For $\mathcal{M}:A\rightarrow B$, $\mathcal{N}:A\rightarrow B$, and $\alpha\in [\frac{1}{2},1)\cup(1,\infty)$, the stabilized {\Renyi} channel divergence is defined as 
\begin{align}\label{eq:stablized channel divergence}
	D_\alpha(\mathcal{M}||\mathcal{N})\defvar\sup_{\omega\in\dop{=}(AR)}D_\alpha\left(\mathcal{M}(\omega_{AR})\middle\vert\middle\vert\mathcal{N}(\omega_{AR})\right),
\end{align}
where $R$ is isomorphic to $A$.
\end{definition}

While we state the following definitions of {\Renyi} conditional entropies and mutual information for subnormalized states, we will in fact not need the subnormalized cases for any $\alpha$ value other than $\alpha=\infty$, which we use when defining the smoothed versions.
\begin{definition}\label{def:sandwiched entropy and info}
({\Renyi} conditional entropies and mutual information)
For any bipartite state $\rho\in
\dop{\leq}(AB)
$ with $\tr{\rho}\neq0$, and $\alpha\in[0,\infty]$, we define the following (sandwiched) {\Renyi} conditional entropies and mutual information:
\begin{align}
    \begin{aligned}
    &H_\alpha(A|B)_\rho\defvar-D_\alpha(\rho_{AB}||\id_A\otimes\rho_B)\\
    &H_\alpha^\uparrow(A|B)_\rho \defvar -\inf_{\sigma_B\in\dop{\leq}(B)}D_\alpha(\rho_{AB}||\id_A\otimes\sigma_B) , \\
    &I_\alpha(A:B)_\rho \defvar D_\alpha(\rho_{AB}||\rho_A\otimes\rho_B),\\
    &\IDA_\alpha(A;B)_\rho \defvar \inf_{\sigma_B\in\dop{\leq}(B)} D_\alpha(\rho_{AB}||\rho_A\otimes\sigma_B) ,\\
    &\IDDA_\alpha(A:B)_\rho \defvar \inf_{\omega_A\in\dop{\leq}(A),\sigma_B\in\dop{\leq}(B)} D_\alpha(\rho_{AB}||\omega_A\otimes\sigma_B).
    \end{aligned}
\end{align}
All the above optimizations can instead be restricted to normalized states without changing the optimal values.
For $\alpha=1$ and $\tr{\rho}=1$, we have $H_\alpha(A|B)=H^\uparrow_\alpha(A|B)$ and their values are equal to the von Neumann conditional entropy; similarly, $I_\alpha(A:B)=\IDA_\alpha(A;B)=\IDDA_\alpha(A:B)$ and their values are equal to the von Neumann mutual information. 
Note that while $I_\alpha$ and $\IDDA_\alpha$ are symmetric with respect to partial registers, $\IDA_\alpha$ is not (i.e., $\IDA_\alpha(A;B)_\rho\neq\IDA_\alpha(B;A)_\rho$).
\end{definition}

From the definitions, it is clear that $H_\alpha \leq H^\uparrow_\alpha$ and $I_\alpha \geq \IDA_\alpha \geq \IDDA_\alpha$, as suggested by the notation. The fact that the optimizations can be restricted to normalized states follows by observing that the definition of $D_\alpha$ implies $D_\alpha(\rho||t\sigma) = D_\alpha(\rho||\sigma) - \log t > D_\alpha(\rho||\sigma)$ for any $t\in(0,1)$,
so given any subnormalized feasible point, one can always obtain a ``better'' feasible point by normalizing it.

We use the quantity $\IDA_\alpha$ to state the following choice of definition for {\Renyi} conditional mutual information, as a difference between {\Renyi} mutual information terms. However, we note that this definition may not be as ``well-behaved'' as its von Neumann analogue; for instance, it is unclear whether it satisfies data-processing on $A$. Still, we introduce it because some of our results are most cleanly stated in terms of this definition.
\begin{definition}\label{def:sandwiched CQMI}
	({\Renyi} conditional mutual information)
	For any tripartite state $\rho\in\dop{=}(ABC)$, and $\alpha\in [0,\infty]$, we define the following (sandwiched) \Renyi\ conditional mutual information
	\begin{align}
		\label{eq:cqmi}
		\Idiff_\alpha(A;B|C)_\rho\defvar\IDA_\alpha(A;BC)_\rho-\IDA_\alpha(A;C)_\rho
	\end{align}
	For $\alpha=1$, this definition is equal to the von Neumann conditional mutual information. Note that this definition is not symmetric with respect to partial registers(i.e., $\Idiff_\alpha(A;B|C)_\rho\neq \Idiff_\alpha(B;A|C)_\rho $).
	\end{definition}

We now list some smoothed entropic quantities and briefly discuss some properties.

\newcommand{\gF}{F} 
\begin{definition}\label{def:genfid}
For $\rho,\sigma \in \dop{\leq}(A)$, the \term{generalized fidelity} is
\begin{align}
\gF(\rho,\sigma) \defvar \norm{\sqrt{\rho}\sqrt{\sigma}}_1 + \sqrt{(1-\tr{\rho})(1-\tr{\sigma})},
\end{align}
and the \term{purified distance} is $\pd(\rho,\sigma)\defvar\sqrt{1-\gF(\rho,\sigma)^2}$. 
If either $\rho$ or $\sigma$ is normalized, this reduces to the standard definition of fidelity between normalized states.
\end{definition}

\begin{definition}
For $\rho\in\dop{\leq}(A), \sigma\in\Pos(A)$ with $\tr{\rho}\neq0$, and any $\epsilon \geq 0$, the $\eps$-\term{smoothed max-divergence} is defined as
\begin{align}
\Dmax^\eps (\rho||\sigma)\defvar 
\inf_{\substack{\tilde{\rho} \in \dop{\leq}(A) \suchthat\\ \pd(\tilde{\rho},\rho)\leq\eps}} \Dmax(\tilde{\rho}||\sigma).
\end{align}
\end{definition}
Note that the infimum is attained in the above definition, but it may not necessarily be attained by a normalized state, even if $\rho$ is normalized~\cite{Tom16}.

Similarly, we can define various versions of smoothed min-entropy or smoothed max-information, corresponding to the assorted possible definitions of {\Renyi} entropies or mutual information for $\alpha=\infty$. In this work, for all our formal definitions and results we use the notations $\Hminup[,\eps]$, $\Imaxnone[\eps]$, $\ImaxDA[,\eps]$, $\ImaxDDA[,\eps]$ (rather than e.g.~$H_\mathrm{min}^{\eps}$ or $I_\mathrm{max}^{\eps}$) for greater clarity in terms of their relations to the above definitions of {\Renyi} conditional entropy and mutual information.
We remark that for smoothed min-entropy, we only introduce the $H^\uparrow_\infty$ version because we do not need the $H_\infty$ version; in any case the term ``min-entropy'' is more commonly used to refer to the former version.

\begin{definition}\label{def:smoothing} (Smoothed min-entropy and smoothed max-information)
For $\rho\in\dop{\leq}(AB)$ with $\tr{\rho}\neq0$, and any $\eps\in\left[0,\sqrt{\tr{\rho_{AB}}}\right)$, let
\begin{align}
\begin{aligned}
&\Hminup[,\eps](A|B)_\rho \defvar
\sup_
{\substack{\tilde{\rho} \in \dop{\leq}(AB) \suchthat\\ \pd(\tilde{\rho},\rho)\leq\eps}}
\Hminup(A|B)_{\tilde{\rho}} = -\inf_{\sigma_B\in\dop{=}(B)}\Dmax^\eps(\rho_{AB}||\id_A\otimes\sigma_B), \\
&\Imaxnone[\eps](A:B)_\rho \defvar
\inf_
{\substack{\tilde{\rho} \in \dop{\leq}(AB) \suchthat\\ \pd(\tilde{\rho},\rho)\leq\eps}} 
\Imaxnone(A:B)_{\tilde{\rho}}, \\
&\ImaxDA[,\eps](A;B)_\rho \defvar
\inf_
{\substack{\tilde{\rho} \in \dop{\leq}(AB) \suchthat\\ \pd(\tilde{\rho},\rho)\leq\eps}} 
\ImaxDA(A;B)_{\tilde{\rho}}, \\
&\ImaxDDA[,\eps](A:B)_\rho \defvar
\inf_
{\substack{\tilde{\rho} \in \dop{\leq}(AB) \suchthat\\ \pd(\tilde{\rho},\rho)\leq\eps}} 
\ImaxDDA(A:B)_{\tilde{\rho}} = \inf_{\omega_A\in\dop{=}(A),\sigma_B\in\dop{=}(B)} \Dmax^\eps(\rho_{AB}||\omega_A\otimes\sigma_B).
\end{aligned}
\end{align}
\end{definition}
Again, it may not be the case that the optimum in the above definitions is attained by a normalized $\tilde{\rho}$, though we know of two special cases. First, for $\Hminup[,\eps](A|B)_\rho$ with normalized $\rho$, the optimum is attainable in some sense by a normalized $\tilde{\rho}$ embedded in a larger Hilbert space~\cite[Lemma~6.5]{Tom16}. Second, for $\ImaxDA[,\eps](A;B)_\rho$ with normalized $\rho$, we observe that the optimum is always attainable by a normalized $\tilde{\rho}$, because the objective function $\ImaxDA(A;B)_{\tilde{\rho}}$ is ``scaling-invariant'': we have $\Dmax(t\tilde{\rho}_{AB}||t\tilde{\rho}_A\otimes\sigma_B) = \Dmax(\tilde{\rho}_{AB}||\tilde{\rho}_A\otimes\sigma_B)$ for any $t\geq 0$ (and any $\sigma$), therefore we can just normalize any $\tilde{\rho}$ attaining the optimum (also noting that since $\rho$ is normalized, this normalization will only increase the value of $\gF(\tilde{\rho},\rho)$ according to Def.~\ref{def:genfid}). In summary, we have
\begin{align}\label{eq:specialImaxeps}
\ImaxDA[,\eps](A;B)_\rho =
\inf_
{\substack{\tilde{\rho} \in \dop{=}(AB) \suchthat\\ \pd(\tilde{\rho},\rho)\leq\eps}} 
\ImaxDA(A;B)_{\tilde{\rho}}.
\end{align} 
We will use this property of $\ImaxDA[,\eps](A;B)_\rho$ in one of our later proofs.

We also remark that in the above definitions, we were able to rewrite $\Hminup[,\eps]$ and $\ImaxDDA[,\eps]$ in terms of $\Dmax^\eps$ by using the fact that infima commute. There is a subtle obstruction to writing similar expressions for $\Imaxnone[\eps]$ and $\ImaxDA[,\eps]$, namely that the smoothing in those cases involves some terms in the second argument of the corresponding $\Dmax$ terms, in contrast to the definition of $\Dmax^\eps$ which only smooths over the first argument. One possible approach to avoid this issue is to instead consider \term{partially smoothed} definitions as introduced in~\cite{ABJT20}, which constrain the partial states in the second argument of the $\Dmax$ terms. However, for this work we shall simply rely on the fact that the various definitions of smoothed max-information can be converted to each other (with some losses) using \cite[Theorem~3]{CBR14}.

We now also list some useful properties we will use throughout our work.

\begin{fact} \label{fact:DPI}
(Data-processing~\cite[Theorem~1]{FL13}; see also~\cite{MDS+13,Beigi13,MO14,Tom16}) For any $\alpha\in[1/2,\infty]$, any $\rho,\sigma\in\Pos(Q)$ with $\tr{\rho}\neq0$, and any channel $\mathcal{E}:Q\to Q'$, we have:
\begin{align}
D_\alpha(\rho\Vert\sigma) \geq D_\alpha(\mathcal{E}[\rho]\Vert\mathcal{E}[\sigma]),
\end{align}
and thus also for any $\rho\in\dop{=}(Q''Q)$,
\begin{align}
\begin{gathered}
H_\alpha(Q''|Q)_{\rho} \leq H_\alpha(Q''|Q')_{\mathcal{E}[\rho]}, \quad H^\uparrow_\alpha(Q''|Q)_{\rho} \leq H^\uparrow_\alpha(Q''|Q')_{\mathcal{E}[\rho]} \\ 
I_\alpha(Q'':Q)_\rho \geq I_\alpha(Q'':Q')_{\mathcal{E}[\rho]}, \quad
\IDDA_\alpha(Q'':Q)_\rho \geq \IDDA_\alpha(Q'':Q')_{\mathcal{E}[\rho]}, \\
\IDA_\alpha(Q'';Q)_\rho \geq \IDA_\alpha(Q'';Q')_{\mathcal{E}[\rho]}, \quad\IDA_\alpha(Q;Q'')_\rho \geq \IDA_\alpha(Q';Q'')_{\mathcal{E}[\rho]}.
\end{gathered}
\end{align}
If $\mathcal{E}$ is an isometry, all the above bounds hold with equality.
\end{fact}

\begin{fact} \label{fact:classmix}
(Conditioning on classical registers; see~\cite[Sec.~III.B.2 and Proposition~9]{MDS+13} or~\cite[Eq.~(5.32) and Proposition~5.1]{Tom16}) Let $\rho,\sigma \in \dop{=}(Q)$ be states with a direct-sum structure $\rho_Q = \bigoplus_{c}\rho(c)\rho_{Q|c}, \sigma_Q = \bigoplus_{c}\sigma(c)\sigma_{Q|c}$ (with the direct sum having the same ordering in both cases), for some probability distributions $\bsym{\rho},\bsym{\sigma}$ and some normalized states $\rho_{Q|c},\sigma_{Q|c}$ indexed by $c$.\footnote{This direct-sum structure automatically holds whenever the state explicitly includes a classical register $C$ that encodes the value of $c$ in the direct sum, but we state it in this form for greater flexibility in applications.} Then
\begin{align}
    \label{eq:classmixD}
    D_\alpha(\rho||\sigma)=\frac{1}{\alpha-1}\log\left(\sum_{c}\rho(c)^\alpha\sigma(c)^{1-\alpha}2^{(\alpha-1)D_\alpha(\rho_{Q|c}||\sigma_{Q|c})}\right).
\end{align}
Hence for a state $\rho \in \dop{=}(C Q Q')$ classical on $C$, we have
\begin{align}
    \label{eq:classmixHdown}
    &H_\alpha(Q|CQ')=\frac{1}{1-\alpha}\log\left(\sum_{c}\rho(c)2^{(1-\alpha)H_\alpha(Q|Q')_{\rho|c}}\right),\\
    \label{eq:classmixHup}
    &H^\uparrow_\alpha(Q|CQ')=\frac{\alpha}{1-\alpha}\log\left(\sum_{c}\rho(c)2^{\frac{1-\alpha}{\alpha}H^\uparrow_\alpha(Q|Q')_{\rho|c}}\right).
\end{align}
\end{fact}
\begin{fact}
	\label{fact:smooth to renyi}
	(Relations between smooth divergence and sandwiched \Renyi\ divergence; Proposition~6.5 in~\cite{Tom16}) Let $\rho\in\dop{=}(A)$, $\sigma\in\Pos(A)$. Then, for any $\epsilon\in (0,1),\ \alpha\in (1,\infty)$, we have
	\begin{align}\label{eq:Dsmooth to renyi}
		\Dmax^\epsilon(\rho||\sigma)\leq D_\alpha(\rho||\sigma)+\frac{g(\epsilon)}{\alpha-1},
	\end{align}
	where $g(\epsilon)\defvar\log\left(1-\sqrt{1-\epsilon^2}\right)$. In particular, for any $\rho\in\dop{=}(AB)$,
	\begin{align}\label{eq:smooth to renyi QMI}
		\ImaxDDA[,\epsilon](A:B)_\rho\leq\IDDA_\alpha(A:B)_\rho+\frac{g(\epsilon)}{\alpha-1}.
	\end{align}
\end{fact}
For completeness, we present how to obtain the second bound from the first (which was proven in~\cite[Proposition~6.5]{Tom16}): observe that for any $\rho\in\dop{=}(AB),\omega_A\in\dop{=}(A),\sigma_B\in\dop{=}(B)$, from the first bound we have
\begin{align}
	\Dmax^\eps(\rho_{AB}||\omega_A\otimes\sigma_B)\leq D_\alpha(\rho_{AB}||\omega_A\otimes\sigma_B)+\frac{g(\eps)}{\alpha-1}.
\end{align}
Taking infimum over $\omega,\sigma$ in the above inequality, and noting the relation between $\ImaxDDA[,\epsilon]$ and $\Dmax^\epsilon$ described in Def.~\ref{def:smoothing}, the claim follows. (Generalizing this property to, say, $\ImaxDA[,\epsilon]$ would require a suitable handling of the reduced state on $A$ in the second argument, as discussed above.)

\section{One-shot chain rule}
\label{sec:Imax_chain}
In this section, we derive a chain rule between smooth min-entropy and smooth max-information, which can be used for one-shot characterization of information leakage caused by an additional conditioning register.
This can be tighter than previous chain rules which were either dimension-based~\cite{Tom16}, or were based on smooth max-entropy~\cite{VDT13}.

We do so via an ``operational'' argument, by drawing a connection to state redistribution protocols. Informally, the idea is that if we can ``reversibly'' compress a quantum register down to some smaller dimension, while in some sense maintaining its correlation with other registers, then the information it contains about other registers should be characterized by only the smaller compressed dimension. To formalize this better, we first define state redistribution protocols:

\begin{definition}
	\label{def:state redistribution}
	($\delta$-Quantum state redistribution; see~\cite{DY08,ADJ17})
	Consider any state $\rho\in\dop{=}(ABCD)$.
	The protocol consists of three parties, Alice, Bob and a reference party, where Alice begins with registers $AC$, Bob has $B$, and the reference party has $D$. In addition, Alice and Bob have access to any shared entangled state $\rho_{A'B'}$. The goal is for Alice to transmit the $C$ register to Bob, possibly assisted by the shared entanglement, such that the final state on $ABCD$ after the protocol is $\delta$-close in purified distance to the initial state. More formally, let $Q$ denote a quantum register that Alice sends to Bob in the protocol. A \term{$\delta$-quantum state redistribution protocol for $\rho$} consists of a pair of encoding and decoding channels $\mathcal{E}:AA'C\rightarrow AQ$ and $\mathcal{D}:QBB'\rightarrow BC$, such that
	\begin{align}
		\pd\left(\rho_{ABCD},\left(\mathcal{D}\circ\mathcal{E}\right)(\rho_{ABCD}\otimes\rho_{A'B'})\right)\leq\delta.
	\end{align}
	We refer to $\log|Q|$ as the \term{quantum communication cost} of the protocol.
\end{definition}
Note that in the above definition, we have not imposed a restriction that $\rho_{ABCD}$ is pure; however, various existing results for state redistribution are only valid for pure states, so some care is needed when applying those results.
Also, while we have allowed the ``message'' register $Q$ to be quantum, it is a standard result that one can restrict to classical messages by simply doubling the communication cost, because it is an entanglement-assisted scenario and hence the pre-shared entanglement can be used to teleport any one qubit with two classical bits. (For classical $Q$, our chain rule below would just have $\log|Q|$ instead of $2\log|Q|$, consistent with this notion that the classical communication cost is twice the quantum communication cost.)

We now use this concept to derive a one-shot chain rule. While the register ordering for the protocol described in this theorem might appear unintuitive at first, we discuss below how in principle it gives the ``correct'' register ordering for a CQMI-like term to quantify the leakage, even if that term is not necessarily easy to analyze.
\begin{theorem}
	\label{thrm:Hmin_Imax chain}
	For any $\rho\in\dop{=}(SLE)$, consider a purification $\ketbra{\rho}{\rho}\in\dop{=}(SLEP)$.
	Suppose there exists a $\delta$-quantum state redistribution protocol for $\rho$ (Def.~\ref{def:state redistribution}) in which Alice holds $PL$, Bob's initial register is $E$, the reference holds $S$, and the register to be transmitted is $L$. Let the quantum communication cost of the protocol be $\log|Q|$ with $Q$ being the encoded register as mentioned in that definition. Then for any $\epsilon\in [0,1)$, the following holds:
	\begin{align}
		\label{eq:Hmin_Imax chain}
		\Hminup[,\eps+\delta](S|LE)_{\rho}\geq \Hminup[,\eps](S|E)_{\rho}-2\log|Q|.
	\end{align}
	If the register $Q$ in the state redistribution protocol is classical, then in the above bound, $2\log|Q|$ can be replaced with $\log|Q|$.
	\begin{proof}
		Consider a $\delta$-quantum state redistribution protocol where Alice, Bob and the reference hold onto the registers indicated in the statement of the theorem. We write the initial state of the protocol as $\rho_{SLEPXY}=\ketbra{\rho}{\rho}_{SLEP}\otimes\rho_{XY}$, where $\rho_{XY}$ is the shared entanglement between Alice and Bob. Then, by Def.~\ref{def:state redistribution} there exist an encoding $\mathcal{E}:LPX\rightarrow PQ$, and decoding $\mathcal{D}:QEY\rightarrow LE$, such that  
		\begin{align}
			\label{eq:EncDec}
			\pd\left(\rho''_{SLEP},\rho_{SLEP}\right)\leq\delta
		\end{align}
		where (leaving identity maps implicit) $\rho''_{SLEP}=\left(\mathcal{D}\circ\mathcal{E}\right)(\rho_{SLEPXY})$. Denoting the state after the encoding map as $\rho'_{SEPQY}=\mathcal{E}(\rho_{SLEPXY})$, we have
		\begin{align}
			\Hminup[,\eps+\delta](S|LE)_{\rho}&\geq \Hminup[,\eps](S|LE)_{\mathcal{D}\circ\mathcal{E}\left[\rho\right]}\notag\\
			&\geq \Hminup[,\eps](S|QEY)_{\mathcal{E}\left[\rho\right]}\notag\\
			&\geq \Hminup[,\eps](S|EY)_{\mathcal{E}\left[\rho\right]}-2\log|Q|\notag\\
			&= \Hminup[,\eps](S|EY)_{\rho}-2\log|Q|\notag\\
			&\geq \Hminup[,\eps](S|E)_{\rho}-2\log|Q|,
		\end{align}
		where the first line holds by Eq.~\eqref{eq:EncDec}, the second line is an application of data-processing (Fact~\ref{fact:DPI}), the third line follows from Prop.~3.3.4 in~\cite{Led16}, the fourth line follows from the fact that the encoding channel $\mathcal{E}$ acts as identity on $SEY$ registers, and the last line holds since by construction $\rho_{SEY}=\rho_{SE}\otimes\rho_Y$. Note that by isometric equivalence of purifications and the isometric invariance property of min-entropies, the above proof holds for all purifications of $\rho_{SLE}$. 
		
		The version for classical $Q$ holds by simply noting that we instead have $\Hminup[,\eps](S|QEY)_{\mathcal{E}\left[\rho\right]}\geq \Hminup[,\eps](S|EY)_{\mathcal{E}\left[\rho\right]}-\log|Q|$ in that case, or alternatively by noting that any such state redistribution protocol can be converted to one where $Q$ is fully quantum but halved in size, via dense coding (since we are allowing arbitrary entanglement assistance).
	\end{proof}
\end{theorem}

\begin{remark}
Due to the ``operational'' nature of our above proof, it straightforwardly generalizes to any other conditional entropy that satisfies data-processing (on the conditioning registers), the ``crude'' log-dimension chain rule\footnote{In fact it suffices to only have such a chain rule hold for classical conditioning registers (subtracting $\log|Q|$ instead), by simply restricting $Q$ in the above proof to be classical and doubling its size using the teleportation argument discussed earlier.} in the third line of our above calculation, and additivity across tensor-product states. For instance, this means that if we define a smoothed version $H_\alpha^{\uparrow,\eps}$ of $H_\alpha^\uparrow$ in the same way as the smoothed min-entropy definition, we immediately get the analogous chain rule
\begin{align}\label{eq:Halpha_Imax chain}
H_\alpha^{\uparrow,\delta}(S|LE)_{\rho}\geq 		H_\alpha^{\uparrow}(S|E)_{\rho}-2\log|Q|.
\end{align}
This may be useful as an alternative to the above; for instance, entropy accumulation theorems~\cite{DFR20,MFSR22} often essentially provide a bound on $H_\alpha^{\uparrow}(S|E)_{\rho}$ that is in some senses tighter than a bound on $\Hminup[,\eps](S|E)_{\rho}$, and hence Eq.~\eqref{eq:Halpha_Imax chain} would be more useful in such contexts. (The smoothing on the left-hand-side is easily handled by deriving a smoothed version of the {\Renyi} privacy amplification theorem of~\cite{Dup21}, by the same arguments as the standard proofs of privacy amplification with smoothed min-entropy.)
\end{remark}

Existing achievability and converse bounds~\cite{AWA18,ARS18,ABRNT23,ADJ17} on one-shot state redistribution protocols (for pure $\rho_{SLEP}$) suggest that the optimal quantum communication cost is characterized by some form of one-shot ``CQMI-like'' quantity, i.e.~informally speaking, $2\log|Q| \approx I_\mathrm{max}^\delta(S;L|E)$ for some notion of $I_\mathrm{max}^\delta(S;L|E)$. However, the quantities that appear in those results are often fairly complicated and not easy to analyze. In general, max-CQMI seems difficult to usefully define: for instance, even though in~\cite{BSW15} there were some attempts to define a non-smoothed version of this quantity, it is not very well-behaved. For this reason, we now present a simplified version of the above result, in which we resort to a looser term based on smooth max-information.

In particular,~\cite{ADJ17} proved an achievable rate for the quantum communication cost of a $\delta$-quantum state redistribution task, involving only the smooth max-information between suitable registers. We use that result to gain a simple bound from Theorem~\ref{thrm:Hmin_Imax chain}, as follows:
\begin{corollary}\label{cor:Hmin_Imax chain}
	For any $\rho\in\dop{=}(SLE)$, $\delta\in(0,1)$, and any $\epsilon\in[0,1)$ the following holds:
	\begin{align}
		\label{eq:Hmin_Imax chain2}
		\Hminup[,\eps+\delta](S|LE)_\rho\geq \Hminup[,\eps](S|E)_\rho-\ImaxDA[,\delta](SE;L)_\rho-\log\left(\frac{4}{\delta^2}\right)
	\end{align}
	\begin{proof}
		According to the result in~\cite{ADJ17}, there exists a $\delta$-quantum state redistribution protocol in which Alice holds $PL$, Bob's registers are trivial\footnote{This scenario is usually referred to as \term{quantum state splitting}, though it can just be viewed as a special case of quantum state redistribution. The ``looseness'' of our resulting bound (in that we have a smooth max-information instead of a ``one-shot CQMI'') can be viewed as the fact that in this special case, we are constraining the operations Bob can perform, and hence this result is suboptimal when considering the broader context of state redistribution rather than the special case of state splitting.}, the reference is $ES$, and the register to be transmitted is $L$, such that the quantum communication cost is\footnote{Strictly speaking the proof in that work only straightforwardly applies when the smoothing in $\ImaxDA[,\delta]$ is taken over normalized states, but as discussed above Eq.~\eqref{eq:specialImaxeps}, this is not an issue for this version of smooth max-information.} $\log |Q|=\frac{1}{2}\ImaxDA[,\delta](SE;L)_\rho+\log\left(\frac{2}{\delta}\right)$. Note however that this can also be validly viewed as a $\delta$-quantum state redistribution protocol where Alice holds $PL$, Bob holds $E$, and the reference is $S$; in which we just constrained Bob's decoding channel to act as the identity on $E$. Therefore, we can apply Theorem~\ref{thrm:Hmin_Imax chain} on this result, which gives
		\begin{align}
			\Hminup[,\eps+\delta](S|LE)_{\rho}&\geq \Hminup[,\eps](S|E)_{\rho}-2\log|Q|\notag\\
			&\geq \Hminup[,\eps](S|E)_{\rho} - \ImaxDA[,\delta](SE;L)_\rho-\log\left(\frac{4}{\delta^2}\right),
		\end{align}
		as claimed.
	\end{proof}
\end{corollary}
The above chain rule is conceptually a tighter bound (apart from possible ``finite-size effects'' regarding the $\eps$ and $\delta$ dependencies) than the previous result of~\cite{VDT13} based on smooth max-entropy, in that it is based on a more natural quantity of smooth max-information, which measures the correlation between the extra conditioning register and the rest of the registers. 

The chain rule presented in Corollary~\ref{cor:Hmin_Imax chain} can be used in cryptography scenarios where we are dealing with device imperfections. One way to employ the above result in a QKD protocol is to set $S$ be the raw secret data, $E$ to be Eve's side information, and $L$ to characterize the leakage caused by having imperfections in the measurement devices (over the entire course of the protocol). In such a scenario, when removing conditioning on the leakage register, naturally one would expect to have a penalty term that measures the amount of leaked information. This expectation is met by the chain rule in  Eq.~\eqref{eq:Hmin_Imax chain2}.
Another possible application would be to address issues regarding protocol composability~\cite{arx_PR21,BCK13}; we discuss this further in the conclusion section.

In fact, it is also possible to prove a chain rule of the form in Corollary~\ref{cor:Hmin_Imax chain} directly from the SDP characterization of min-entropy~\cite{Tom16}, although with $\Imaxnone[\eps]$ instead of $\ImaxDA[,\eps]$. First, we present a version without smoothing:
\begin{theorem} \label{thm:chain_rule_non_smoothed}
For any $\rho\in\dop{\leq}(SLE)$, the following holds:
\begin{equation}
\Hminup(S|LE)_{\rho} \geq \Hminup(S|E)_{\rho} - \Imaxnone(L:SE)_{\rho} \,.
\end{equation}
\end{theorem}
\begin{proof}
Consider the SDP for $\Dmax$~\cite[Chapter~4.2.4]{Tom16}:
\begin{equation}
\begin{aligned}
2^{\Dmax(\rho_{SLE} || \rho_L \otimes \rho_{SE})} = 
& \underset{\mu}{\text{minimize}} & & \mu \\
& \text{subject to} & & \rho_{SLE} \leq \mu \rho_L \otimes \rho_{SE} \\
&&& \mu \geq 0 \,.
\end{aligned}
\end{equation}
From this, it is clear that 
\begin{equation}
\rho_{SLE} \leq 2^{\Dmax(\rho_{SLE} || \rho_L \otimes \rho_{SE})} \rho_L \otimes \rho_{SE} \,. \label{eqn:dmax_bound}
\end{equation}
Now consider the SDP for $\Hminup$~\cite[Definition~6.2]{Tom16} (note that this is consistent with the $\Hminup$ definition we used; Definition~\ref{def:sandwiched entropy and info}): 
\begin{align*}
2^{-\Hminup(S|LE)_{\rho}} & = & \underset{\sigma_{LE}}{\text{minimize}}
&& &  \tr{\sigma_{LE}} \\
&& \text{subject to}
&& & \id_S \otimes \sigma_{LE} \geq \rho_{SLE} \\
&&&&& \sigma_{LE} \geq 0 \\
&\leq & \underset{\sigma_{LE}}{\text{minimize}}
&& &  \tr{\sigma_{LE}} \numberthis\label{eqn:hmin_sdp} \\
&& \text{subject to}
&& & \id_S \otimes \sigma_{LE} \geq 2^{\Dmax(\rho_{SLE} || \rho_L \otimes \rho_{SE})} \rho_L \otimes \rho_{SE} \\
&&&&& \sigma_{LE} \geq 0 \,.
\end{align*}
The inequality holds because by Eq.~\eqref{eqn:dmax_bound}, the feasible set in the second SDP is smaller.

We further consider
\begin{equation}
\begin{aligned}
2^{-\Hminup(S|E)_{\rho}} = & \underset{\omega_E}{\text{minimize}}
& &  \tr{\omega_E} \\
& \text{subject to}
& & \id_S \otimes \omega_E \geq \rho_{SE} \\
&&& \omega_E \geq 0 \,,
\end{aligned}
\end{equation}
and denote an optimizer for this by $\omega_E^*$, i.e. $\omega_E^*$ is feasible and $\tr{\omega_E^*} = 2^{-\Hminup(S|E)_{\rho}}$.

For the SDP in Eq.~\eqref{eqn:hmin_sdp}, we now pick 
\begin{equation}
\sigma_{LE} = 2^{\Dmax(\rho_{SLE} || \rho_L \otimes \rho_{SE})} \rho_L \otimes \omega_E^* \,. 
\end{equation}
This choice of $\sigma_{LE}$ is feasible since  
\begin{align}
\id_S \otimes \sigma_{LE} = 2^{\Dmax(\rho_{SLE} || \rho_L \otimes \rho_{SE})} \id_S \otimes \rho_L \otimes \omega_E^* \geq 2^{\Dmax(\rho_{SLE} || \rho_L \otimes \rho_{SE})} \rho_L \otimes \rho_{SE} \,, 
\end{align}
where the inequality holds because $\omega_E^*$ is feasible, i.e., $\id_S \otimes \omega_E^* \geq \rho_{SE}$. Therefore we get 
\begin{align}
2^{-\Hminup(S|LE)_{\rho_{SLE}}} &\leq \tr{2^{\Dmax(\rho_{SLE} || \rho_L \otimes \rho_{SE})} \rho_L \otimes \omega_E^*} \notag\\
&\leq 2^{\Dmax(\rho_{SLE} || \rho_L \otimes \rho_{SE})} \cdot \tr{\omega_E^*} \notag\\
&= 2^{\Dmax(\rho_{SLE} || \rho_L \otimes \rho_{SE}) -\Hminup(S|E)_{\rho} } \,.
\end{align}\
Taking the logarithm yields the result.
\end{proof}

\newcommand{\sth}{{\text{~s.t.~}}}
\newcommand{\optstate}{\hat{\omega}}
We now use this to prove a chain rule for the corresponding smoothed quantities.
\begin{corollary} \label{cor:chain_rule_smoothed}
For any $\rho\in\dop{\leq}(SLE)$ and $\eps \in [0,1)$, the following holds:
\begin{equation}
\Hminup[,\eps](S|LE)_{\rho} \geq \Hminup[,\eps/3](S|E)_{\rho} - \Imaxnone[\eps/3](L:SE)_\rho \,.
\end{equation}
\end{corollary}
\begin{proof}
The proof technique we use here is similar to the one used in \cite{CBR14,BCR11,DBWR14}. For compactness of notation, let us write $B_\eps(\rho) \subset \dop{\leq}(SLE)$ to denote the set of subnormalised states $\eps$-close to $\rho$ in purified distance, and $L(R)$ to denote the set of linear operators on some register $R$.
The proof follows from a sequence of lower bounds on $\Hminup[,\eps](X|BC)_{\rho}$. First, from Theorem~\ref{thm:chain_rule_non_smoothed} we have 
\begin{align}
\Hminup[,\eps](S|LE)_{\rho} 
&\geq  \max_{\omega \in B_\eps(\rho)} \left(  \Hminup(S|E)_{\omega} - \Imaxnone(L:SE)_{\omega}  \right) \label{eqn:smoothed1}
\end{align}
We can lower bound the maximization by splitting it into two parts as follows:
\begin{align}
\text{\eqref{eqn:smoothed1}} \geq \max_{\optstate \in B_{\eps/3}(\rho)} \quad 
\max_{
\substack{T_{SE} \in L(SE) \sth \\ T_{SE}\optstate_{SLE}T_{SE}^\dagger \in B_{2\eps/3}(\optstate)}
} \left(  \Hminup(S|E)_{T_{SE}\optstate_{SLE}T_{SE}^\dagger} - \Imaxnone(L:SE)_{T_{SE}\optstate_{SLE}T_{SE}^\dagger}  \right) \label{eqn:smoothed2}
\end{align}
The above is a lower bound because the maximisation in Eq.~\eqref{eqn:smoothed2} is over a subset of the maximisation in Eq.~\eqref{eqn:smoothed1}.
We now observe that since $\Dmax$ satisfies data-processing for any completely positive map~\cite{Tom16}, it follows that $\Imaxnone(L:SE)_{T_{SE}\optstate_{SLE}T_{SE}^\dagger} \leq \Imaxnone(L:SE)_{\optstate_{SLE}}$. Therefore, it is in turn lower bounded by:
\begin{align}
\text{\eqref{eqn:smoothed2}} \geq \max_{\optstate \in B_{\eps/3}(\rho)} \quad \max_{
\substack{T_{SE} \in L(SE) \sth \\ T_{SE}\optstate_{SLE}T_{SE}^\dagger \in B_{2\eps/3}(\optstate)}
} \left(  \Hminup(S|E)_{T_{SE}\optstate_{SLE}T_{SE}^\dagger} - \Imaxnone(L:SE)_{\optstate_{SLE}}  \right) \label{eqn:smoothed3}
\end{align}
We now pick $\optstate \in B_{\eps/3}(\rho)$ such that $\Imaxnone(L:SE)_{\optstate_{SLE}} = \Imaxnone[\eps/3](L:SE)_{\rho}$, obtaining another lower bound:
\begin{align}
\text{\eqref{eqn:smoothed3}} \geq \max_{
\substack{T_{SE} \in L(SE) \sth\\ T_{SE}\optstate_{SLE}T_{SE}^\dagger \in B_{2\eps/3}(\optstate)}
} \left(  \Hminup(S|E)_{T_{SE}\optstate_{SLE}T_{SE}^\dagger}\right) - \Imaxnone[\eps/3](L:SE)_{\rho}
\end{align}
Let $\tilde{\rho}_{SE} \in B_{\eps/3}(\rho_{SE})$ be such that $\Hminup[,\eps/3](S|E)_{\rho} = \Hminup(S|E)_{\tilde{\rho}}$.
By \cite[Lemma B.3]{DBWR14} there exists a $\hat T_{SE}$ such that $\hat T_{SE} \optstate_{SLE} \hat T_{SE}^\dagger \in \dop{\leq}(SLE)$ is an extension of $\tilde{\rho}_{SE}$ and $\pd(\optstate_{SLE}, \hat T_{SE} \optstate_{SLE} \hat T_{SE}^\dagger) = \pd(\optstate_{SE}, \tilde{\rho}_{SE})$.
We can bound the latter quantity using the triangle inequality for the purified distance, obtaining $ \pd(\optstate_{SE}, \tilde{\rho}_{SE}) \leq  \pd(\optstate_{SE}, \rho_{SE}) +  \pd(\rho_{SE}, \tilde{\rho}_{SE}) \leq \eps/3+\eps/3 = 2\eps/3$.
Therefore, this $\hat T_{SE}$ lies inside the set of $T_{SE}$ over which we optimise, so we can pick $\hat T_{SE}$ to get a final lower bound:
\begin{align}
\Hminup(S|E)_{\tilde{\rho}} - \Imaxnone[\eps/3](L:SE)_{\rho}
= \Hminup[,\eps/3](S|E)_{\rho} - \Imaxnone[\eps/3](L:SE)_{\rho} \,,
\end{align}
as claimed.
\end{proof}

\section{Information bounding theorem}
\label{sec:Info_Acc}

We now derive a theorem that bounds the \Renyi\ mutual information produced by a sequence of channels, in a similar sense to the entropy accumulation theorem~\cite{DFR20,MFSR22}, except that here we instead have an upper bound (which in some sense cannot have a ``matching'' lower bound; see Remark~\ref{remark:MMI} later).
To achieve this, we start with the following lemma which can be used to establish a chain rule for \Renyi\ mutual information.

\begin{lemma}\label{lemma:D_chain}
	Let $\rho\in\dop{=}(ABC)$, $\sigma\in\Pos(C)$, $\omega'\in\dop{=}(A)$, and $\omega''\in\dop{=}(B)$. Then, for any $\alpha\in [0,\infty)$ we have
\begin{align}
	\label{eq:D_chain}
	\inf_{\omega_{AB}\in\dop{=}(AB)} D_\alpha(\rho_{ABC}||\omega_{AB}\otimes\sigma_C)\leq D_\alpha(\rho_{AC}||\omega'_{A}\otimes\sigma_C)+D_\alpha(\nu_{ABC}||\omega''_{B}\otimes\nu_{AC}),
\end{align}
where $\nu_{ABC}=\nu_{AC}^{\frac{1}{2}}\rho_{B|AC}\nu_{AC}^{\frac{1}{2}}$, with
\begin{align}
	\label{eq:nu}
	\nu_{AC}=\frac{\left(\rho_{AC}^{\frac{1}{2}}\left(\omega'_A\otimes\sigma_C\right)^{\frac{1-\alpha}{\alpha}}\rho_{AC}^{\frac{1}{2}}\right)^\alpha}{\Tr\left[\left(\rho_{AC}^{\frac{1}{2}}\left(\omega'_A\otimes\sigma_C\right)^{\frac{1-\alpha}{\alpha}}\rho_{AC}^{\frac{1}{2}}\right)^\alpha\right]}.
\end{align}
In particular
\begin{align}
	\label{eq:minfo_acc}
	\IDA_\alpha(C;AB)_\rho\leq\IDA_\alpha(C;A)_\rho+\IDA_\alpha(AC;B)_\nu
\end{align}
\begin{proof}
The proof follows from the following crucial relation:
\begin{align}
	\label{eq:D_chain_equality}
		& D_\alpha(\rho_{ABC}||\omega'_{A}\otimes\omega''_B\otimes\sigma_C)\notag\\
		=&\frac{1}{\alpha-1}\log\Tr\left[\left(\rho_{ABC}\left(\omega'_A\otimes\omega''_B\otimes\sigma_C\right)^\frac{1-\alpha}{\alpha}\right)^\alpha\right]\notag\\
		=&\frac{1}{\alpha-1}\log\Tr\left[\left(\rho_{ABC}\rho_{AC}^{-\frac{1}{2}}\rho_{AC}^{\frac{1}{2}}\left(\omega'_A\otimes\omega''_B\otimes\sigma_C\right)^\frac{1-\alpha}{\alpha}\rho_{AC}^{\frac{1}{2}}\rho_{AC}^{-\frac{1}{2}}\right)^\alpha\right]\notag\\
		=&\frac{1}{\alpha-1}\log\Tr\left[\left(\rho_{ABC}\rho_{AC}^{-\frac{1}{2}}\rho_{AC}^{\frac{1}{2}}\left(\omega'_A\otimes\sigma_C\right)^\frac{1-\alpha}{\alpha}\rho_{AC}^{\frac{1}{2}}\rho_{AC}^{-\frac{1}{2}}\otimes\omega''^{\frac{1-\alpha}{\alpha}}_B\right)^\alpha\right]\notag\\
		=&\frac{1}{\alpha-1}\log\Tr\left[\left(\rho_{ABC}\rho_{AC}^{-\frac{1}{2}}\nu_{AC}^{\frac{1}{\alpha}}\rho_{AC}^{-\frac{1}{2}}\otimes\omega''^{\frac{1-\alpha}{\alpha}}_B\right)^\alpha\right]+\frac{1}{\alpha-1}\log\Tr\left[\left(\rho_{AC}^{\frac{1}{2}}\left(\omega'_A\otimes\sigma_C\right)^{\frac{1-\alpha}{\alpha}}\rho_{AC}^{\frac{1}{2}}\right)^\alpha\right]\notag\\
		=&\frac{1}{\alpha-1}\log\Tr\left[\left(\rho_{B|AC}\nu_{AC}^{\frac{1}{\alpha}}\otimes\omega''^{\frac{1-\alpha}{\alpha}}_B\right)^\alpha\right]+D_\alpha(\rho_{AC}||\omega'_{A}\otimes\sigma_C)\notag\\
		=&\frac{1}{\alpha-1}\log\Tr\left[\left(\nu_{ABC}\left(\nu_{AC}\otimes\omega''_B\right)^{\frac{1-\alpha}{\alpha}}\right)^\alpha\right]+D_\alpha(\rho_{AC}||\omega'_{A}\otimes\sigma_C)\notag\\
		=&D_\alpha(\nu_{ABC}||\omega''_{B}\otimes\nu_{AC})+D_\alpha(\rho_{AC}||\omega'_{A}\otimes\sigma_C),
\end{align}
where the second line follows from the cyclic properties of trace, fourth line is just a reshuffling, fifth line follows from substituting Eq.~\eqref{eq:nu}, the sixth line is a reshuffling and using the definition of conditional state in Def.~\ref{def:cond_state}, and the penultimate line is from definition of $\nu_{ABC}$, and reshuffling terms. 

From Eq.~\eqref{eq:D_chain_equality} we immediately obtain Eq.~\eqref{eq:D_chain} as a consequence, 
by simply picking $\omega_{AB} = \omega'_A \otimes \omega''_B$ as a feasible point in the infimum over $\omega_{AB}$.
Now to prove Eq.~\eqref{eq:minfo_acc}, we pick $\sigma_C=\rho_C$, and $\omega''_B=\omega^*_B, \omega'_A=\omega^*_A$ such that the following relations hold:
 \begin{align}
 	&D_\alpha(\nu_{ABC}||\omega^*_B\otimes\nu_{AC})=\inf_{\sigma_{B}\in\dop{=}(B)}D_\alpha(\nu_{ABC}||\sigma_{B}\otimes\nu_{AC})\notag\\
 	&D_\alpha(\rho_{AC}||\omega^*_A\otimes\rho_{C})=\inf_{\sigma_{A}\in\dop{=}(A)}D_\alpha(\rho_{AC}||\sigma_{A}\otimes\rho_{C}).
 \end{align}
 With these choices, Eq.~\eqref{eq:D_chain} states:
 \begin{align}
 	\IDA_\alpha(C;AB)_\rho&=\inf_{\sigma_{AB}\in\dop{=}(AB)} D_\alpha(\rho_{ABC}||\sigma_{AB}\otimes\rho_C)\notag\\&\leq D_\alpha(\rho_{AC}||\omega^*_A\otimes\rho_{C})+D_\alpha(\nu_{ABC}||\omega^*_B\otimes\nu_{AC})\notag\\
 	&=\IDA_\alpha(C;A)_\rho+\IDA_\alpha(AC;B)_\nu,
 \end{align}
 where the second line follows from definition~\ref{def:sandwiched entropy and info}.
\end{proof}
\end{lemma}

\begin{remark}\label{remark:MMI}
Eq.~\eqref{eq:D_chain_equality} in the above proof (which is an \emph{equality}, not an inequality) can be viewed as a {\Renyi} analogue of the following relation for von Neumann mutual information:
\begin{align}
D(\rho_{ABC}||\rho_{A}\otimes\rho_B\otimes\rho_C) =  D(\rho_{AC}||\rho_{A}\otimes\rho_C)+D(\rho_{ABC}||\rho_{B}\otimes\rho_{AC}) =  I(A:C)+I(AC:B),
\end{align}
where the quantity on the left-hand-side is sometimes termed the \term{multipartite mutual information} $I(A:B:C) \defvar D(\rho_{ABC}||\rho_{A}\otimes\rho_B\otimes\rho_C)$, and is an upper bound on $I(AB:C)$. Hence qualitatively speaking, our above bound is loose up to the difference between the {\Renyi} analogues of $I(A:B:C)$ and $I(AB:C)$ (putting aside the fact that one of the terms in our bound is evaluated on a new state $\nu$). Alternatively, we can view our bound as being a loosened version of the chain rule for von Neumann mutual information:
\begin{align}
I(AB:C) =  I(A:C)+I(A:B|C),
\end{align}
in which the $I(A:B|C)$ term has been relaxed to its upper bound $I(AC:B)$. 

Due to this, it is not straightforward to find a converse to our above bound that would instead prove the {\Renyi} mutual information ``accumulates'', since our bound is inherently somewhat loose compared to the von Neumann versions.
However, despite this looseness, our bounds still suffice for various applications, as we describe later.
\end{remark}

The drawback of Eq.~\eqref{eq:minfo_acc} is that $\IDA_\alpha(AC;B)_\nu$ has to be evaluated over a complicated state; this difficulty can be avoided in that term by replacing it with an optimization over input states to some channel.
\begin{lemma}
	\label{lemma:minfo_acc_relaxed}
	Let $\rho'\in\dop{=}(ARC)$ and $\mathcal{N}:R\rightarrow B$ be a channel. Let $\rho_{ABC}=\mathcal{N}(\rho'_{ARC})$, then for any $\alpha\in [0,\infty)$ we have:
	\begin{align}\label{eq:minfo_acc_relaxed}
		\IDA_\alpha(C;AB)_{\mathcal{N}(\rho')}\leq\IDA_\alpha(C;A)_{\rho'}+\sup_{\omega\in\dop{=}(ARC) }\IDA_\alpha(AC;B)_{\mathcal{N}(\omega)}.
	\end{align}
\end{lemma} 
\begin{proof}
	 Defining  $S_\rho\coloneqq\left\{\sigma\in\dop{=}(ABC)\middle\vert\ \sigma_{B|AC}=\rho_{B|AC}\right\}$, we have
	 \begin{align}
	 	\IDA_\alpha(C;AB)_\rho\leq\IDA_\alpha(C;A)_\rho+\sup_{\sigma\in S_\rho }\IDA_\alpha(AC;B)_{\sigma},
	 \end{align}
	 where the inequality follows from Eq.~\eqref{eq:minfo_acc}, and by noting that $\nu_{ABC}$ in that equation is a feasible solution to the above optimization. Since by definition $\rho'_{AC}=\rho_{AC}$, it is now only needed to show that the optimization in Eq.~\eqref{eq:minfo_acc_relaxed} contains the set $S_\rho$ defined above (i.e., $\nu_{ABC}\in S_\rho$). In order to do so let $\rho_{ABC}=\mathcal{N}(\rho'_{ARC})$, Defining $\omega\in\dop{=}(ARC)$ such that
	 \begin{align}
	 	\omega_{ARC}=\nu_{AC}^{\frac{1}{2}}\rho_{AC}^{-\frac{1}{2}}\rho'_{ARC}\rho_{AC}^{-\frac{1}{2}}\nu_{AC}^{\frac{1}{2}},
	 \end{align}
	 then we have
	 \begin{align}
	 	\mathcal{N}(\omega_{ARC})&=\nu_{AC}^{\frac{1}{2}}\rho_{AC}^{-\frac{1}{2}}\mathcal{N}(\rho'_{ARC})\rho_{AC}^{-\frac{1}{2}}\nu_{AC}^{\frac{1}{2}}\notag\\
	 	&=\nu_{AC}^{\frac{1}{2}}\rho_{B|AC}\nu_{AC}^{\frac{1}{2}}\notag\\
	 	&=\nu_{ABC}.
	 \end{align}
	 where the last line follows from definition of $\nu_{ABC}$ in lemma~\ref{lemma:D_chain}.
\end{proof}

Having proven the above chain rule, we can iteratively apply it to get the information bounding theorem in Eq.~\eqref{eq:introRenyiIAT}, which bounds the \Renyi\ mutual information of a state that is produced by a sequence of channels, in terms of the \Renyi\ mutual information of the individual channel outputs. We can also connect it to the smooth max-information of the final state, for application in the Corollary~\ref{cor:Hmin_Imax chain}--\ref{cor:chain_rule_smoothed} chain rules we described above (after converting to $\ImaxDA[,\eps]$ or $\Imaxnone[\eps]$ via~\cite[Theorem~3]{CBR14}):
\begin{theorem}
	[Information bounding theorem]
	\label{theorem:Info accum}
	Consider a sequence of channels $\left\{\mathcal{M}_i:R_{i-1}\rightarrow R_iL_i\right\}_{i=1}^n$. Let $\epsilon\in[0,1)$. Then, for any $\alpha>1$, and any $\rho\in\dop{=}(R_0SE)$ we have:
	\begin{align}
		\label{eq:Info accum}
		\ImaxDDA[,\epsilon](SE:L_1^n)_{\mathcal{M}_n\circ\cdots\circ\mathcal{M}_1(\rho_{R_0SE})}\leq \sum_{i=1}^n\sup_{\sigma\in\dop{=}(R_{i-1}SE)}\IDA_\alpha(L_1^{i-1}SE;L_i)_{\mathcal{M}_i(\sigma_{R_{i-1}SE})}-\frac{g(\epsilon)}{\alpha-1},
	\end{align}
	where $g(\epsilon)=\log\left(1-\sqrt{1-\epsilon^2}\right)$.
	\begin{proof}
		First note that by combining Fact~\ref{fact:smooth to renyi} with the generic inequality $\IDDA_\alpha \leq \IDA_\alpha$, we have
		\begin{align}\label{eq:Info accum1}
			\ImaxDDA[,\epsilon](SE:L_1^n)_{\mathcal{M}_n\circ\cdots\circ\mathcal{M}_1(\rho_{R_0SE})}&\leq \IDA_\alpha(SE;L_1^n)_{\mathcal{M}_n\circ\cdots\circ\mathcal{M}_1(\rho_{R_0SE})}-\frac{g(\epsilon)}{\alpha-1}.
		\end{align}
		Then to bound the \Renyi\ mutual information term, we employ Lemma~\ref{lemma:minfo_acc_relaxed}:
		\begin{align}
			\IDA_\alpha(SE;L_1^n)_{\mathcal{M}_n\circ\cdots\circ\mathcal{M}_1(\rho_{R_0SE})}\leq \IDA_\alpha(SE;L_1^{n-1})_{\mathcal{M}_{n-1}\circ\cdots\circ\mathcal{M}_1(\rho_{R_0SE})}+\sup_{\sigma\in\dop{=}(R_{n-1}SE)}\IDA_\alpha(L_1^{n-1}SE;L_n)_{\mathcal{M}_n(\sigma_{R_{n-1}SE})}.
		\end{align}
		Repeating this step $n$ times we get:
		\begin{align}\label{eq:Imax vhain proof2}
			\IDA_\alpha(SE;L_1^n)_{\mathcal{M}_n\circ\cdots\circ\mathcal{M}_1(\rho_{R_0SE})}\leq\sum_{i=1}^n\sup_{\sigma\in\dop{=}(R_{i-1}SE)}\IDA_\alpha(L_1^{i-1}SE;L_i)_{\mathcal{M}_i(\sigma_{R_{i-1}SE})}.
		\end{align}
		Combining Eq.~\eqref{eq:Imax vhain proof2} with \eqref{eq:Info accum1} proves the Theorem.
	\end{proof}
\end{theorem}
In principle (assuming the registers have bounded dimension), one could bound the single-round {\Renyi} mutual information terms in terms of von Neumann mutual information via continuity bounds of the form $\IDA_\alpha(L_1^{i-1}SE;L_i) \leq I(L_1^{i-1}SE;L_i) + O(\alpha-1)$ (see~\cite{DFR20,DF19}; the bounds in the latter may be sharper in some contexts). By choosing $\alpha = 1+\Theta(1/\sqrt{n})$, one would then obtain a bound relating smooth max-information to von Neumann mutual information in the form described in Eq.~\eqref{eq:introIAT}. However, we note later in Sec.~\ref{sec:examples} that tighter bounds can be usually obtained by analyzing the {\Renyi} mutual information directly. We also defer the discussion of how to get explicit bounds on these quantities to that section. 

\begin{remark}
Note that another work~\cite{TMM16} has also analyzed {\Renyi} mutual information quantities in some ``sequential'' processes, though mainly those focused on contexts such as quantum hypothesis testing and quantum-feedback-assisted channel capacities. We believe that the proof techniques in that work (specifically, Sec.~5) can be used to obtain a bound with a somewhat similar form to Theorem~\ref{theorem:Info accum}, but where the global state is instead $\mathcal{M}_n \otimes \cdots \otimes \mathcal{M}_1(\rho_{R_1^n SE})$ for some channels  $\left\{\mathcal{M}_i:R_i\rightarrow L_i\right\}_{i=1}^n$; in other words, it is produced by a tensor product of channels rather than a sequence of channels. However, we leave a detailed investigation of this point for future work. We thank M.~Wilde for discussions regarding this work.
\end{remark}

While the approach of combining Corollary~\ref{cor:Hmin_Imax chain}~or~\ref{cor:chain_rule_smoothed} with Theorem~\ref{theorem:Info accum} is fairly straightforward to implement, it has a drawback in that it requires the registers $L_1^n$ to be handled ``separately'' from the other side-information. This implies for instance that an adversary Eve in a protocol cannot ``adapt'' her attack using the leakage register $L_i$ in each round (see the ``restricted adaptiveness'' condition described in~\cite{arx_Tan23}). Also, there may be some subtle conditions in ensuring that the channels describing leakage in the protocol satisfy the conditions of Theorem~\ref{theorem:Info accum}; for instance, it seems we would require all of the secret raw data to be generated first, and then keep a separate register $R_0$ from which the leakage registers are produced.\footnote{One simple approach would be just to have $R_0$ be a classical copy of $S$, but there may be some leakage models in which such a channel structure does not quite work.} In the next section, we describe a different approach that can handle this issue, by having the leakage processes ``built in'' to an entropy accumulation framework.

\section{Entropy accumulation with leakage}\label{sec:GEATleakage}
One drawback of the entropy accumulation framework is that it cannot directly handle some forms of device imperfections, such as photon leakage from the measurement devices, or source correlations.
The reason for this is that to analyze such scenarios, one essentially needs to introduce a leakage register that sends information to the adversary in each round. However, such a construction violates the non-signalling (NS) condition in~\cite{MFSR22}, or the Markov condition in~\cite{DFR20}. In this section we derive a chain rule which can be used to modify the entropy accumulation theorem in a way that allows us to handle device imperfections.

\subsection{Relating {\Renyi} conditional entropies to {\Renyi} mutual information}
\label{sec:HandIbounds}

We begin with the following corollary that states a somewhat similar result to Corollaries~\ref{cor:Hmin_Imax chain}~and~\ref{cor:chain_rule_smoothed}, but instead relates \Renyi\ conditional entropies to \Renyi\ mutual information.
\begin{corollary}
	\label{cor:mutual info chain}
	Let $\rho\in\dop{=}(ABC)$ and $\nu_{ABC}=\nu_{AC}^{\frac{1}{2}}\rho_{B|AC}\nu_{AC}^{\frac{1}{2}}$, where $\nu_{AC}$ is defined in~\eqref{eq:nu}. Then, for any $\alpha\in[0,\infty)$ we have
	\begin{align}
		\label{eq:mutual info chain}
		H_\alpha^\uparrow(C|AB)_\rho\geq H_\alpha^\uparrow(C|A)_\rho -\IDA_\alpha(AC;B)_\nu.
	\end{align}
	Moreover, let $\rho'\in\dop{=}(RAC)$, and $\mathcal{M}:R\rightarrow B$ be a channel such that $\rho_{ABC}=\mathcal{M}(\rho'_{RAC})$. Then, the above relation can be relaxed to the following
	\begin{align}
		\label{eq:mutual info chain_relaxed}
		H_\alpha^\uparrow(C|AB)_{\mathcal{M}(\rho')}\geq H_\alpha^\uparrow(C|A)_{\rho'} -\sup_{\omega\in\dop{=}(RAC)}\IDA_\alpha(AC;B)_{\mathcal{M}(\omega)}.
	\end{align}
	\begin{proof}
		The proof of Eq.\eqref{eq:mutual info chain} follows directly by applying Lemma~\ref{lemma:D_chain}, setting the states in that lemma to be $\sigma_C=\id_C$ and $\omega'_A,\ \omega''_B$ to be such that
		\begin{align}
			&D_\alpha(\rho_{AC}||\omega'_A\otimes\id_C)=\inf_{\sigma_{A}\in\dop{=}(A)}D_\alpha(\rho_{AC}||\sigma_A\otimes\id_C),\notag\\
			&D_\alpha(\nu_{ABC}||\omega''_B\otimes\nu_{AC})=\inf_{\sigma_{B}\in\dop{=}(B)}D_\alpha(\nu_{ABC}||\sigma_B\otimes\nu_{AC}).
		\end{align}
		Then we have
		\begin{align}
			H^\uparrow_\alpha(C|AB)_\rho&=-\inf_{\sigma_{AB}\in\dop{=}(AB)}D_\alpha(\rho_{ABC}||\sigma_{AB}\otimes\id_C)\notag\\
			&\geq -D_\alpha(\rho_{AC}||\omega'_A\otimes\id_C) -D_\alpha(\nu_{ABC}||\omega''_B\otimes\nu_{AC})\notag\\
			&=H^\uparrow_\alpha(C|A)_\rho - \IDA_\alpha(AC;B)_\nu,
		\end{align}
		where the second line follows by applying Lemma~\ref{lemma:D_chain}, and the last line holds by definition. To prove Eq.~\eqref{eq:mutual info chain_relaxed}, let us define $S_\rho=\left\{\sigma\in\dop{=}(ABC)\middle\vert\ \sigma_{B|AC}=\rho_{B|AC}\right\}$. Since by definition $\nu_{B|AC}=\rho_{B|AC}$, we have $\nu_{ABC}\in S_\rho$, and thus we can relax Eq.~\eqref{eq:mutual info chain} to the following:
		\begin{align}
		H_\alpha^\uparrow(C|AB)_\rho\geq H_\alpha^\uparrow(C|A)_\rho -\sup_{\sigma\in S_\rho}\IDA_\alpha(AC;B)_\sigma.			
		\end{align}
		It is now suffices to show that the optimization in Eq.~\eqref{eq:mutual info chain} contains the set $S_\rho$. Defining $\omega\in\dop{=}(RAC)$ as
		\begin{align}
			\omega_{RAC}=\nu_{AC}^\frac{1}{2}\rho_{AC}^{-\frac{1}{2}}\rho'_{RAC}\rho_{AC}^{-\frac{1}{2}}\nu_{AC}^\frac{1}{2}
		\end{align}
		Then we have
		\begin{align}
			\mathcal{M}(\omega_{RAC})&=\nu_{AC}^\frac{1}{2}\rho_{AC}^{-\frac{1}{2}}\mathcal{M}(\rho'_{RAC})\rho_{AC}^{-\frac{1}{2}}\nu_{AC}^\frac{1}{2}\notag\\
			&=\nu_{ABC}.
		\end{align}
		It is clear that $\mathcal{M}(\omega_{RAC})\in S_\rho$, and thus the lemma is proven.
	\end{proof}
\end{corollary}
	The result in Corollary~\ref{cor:mutual info chain} can be tightened further in a certain sense. Intuitively, at least for the limit of $\alpha\rightarrow 1$ (i.e., the von Neumann limit), one would expect that when removing a conditional register from conditional entropy, the penalty term should be a conditional mutual information as opposed to just the unconditional mutual information. In the following, we derive a chain rule with a \Renyi\ version of conditional mutual information. We begin by proving a statement without any conditioning registers on the right-hand-side:
\begin{lemma}
	\label{lemma:C_chain}
	Let $\rho\in\dop{=}(AB)$, $\sigma\in\Pos(B)$, and $\omega\in\Pos(A)$. Then, for any $\alpha\in[0,\infty)$ we have
	\begin{align}
		\label{eq:C_chain}
		D_\alpha(\rho_{AB}||\omega_A\otimes\sigma_B)=D_\alpha(\rho_A||\omega_A)+D_\alpha(\eta_{AB}||\eta_A\otimes\sigma_B),
	\end{align}
	where $\eta_{AB}=\eta_A^{\frac{1}{2}}\rho_{B|A}\eta_A^{\frac{1}{2}}$, with
	\begin{align}
		\label{eq:eta}
		\eta_A=\frac{\left(\rho_A^{\frac{1}{2}}\omega_A^{\frac{1-\alpha}{\alpha}}\rho_A^{\frac{1}{2}}\right)^\alpha}{\Tr\left[\left(\rho_A^{\frac{1}{2}}\omega_A^{\frac{1-\alpha}{\alpha}}\rho_A^{\frac{1}{2}}\right)^\alpha\right]}.
	\end{align}
	In particular 
	\begin{align}
		\label{eq:minfo_1}
		H^\uparrow_\alpha(A|B)_\rho=H_\alpha(A)_\rho-\IDA_\alpha(A;B)_\eta.
	\end{align}
	\begin{proof}
		The proof holds by the following chain of equalities:
		\begin{align}
			\label{eq:C_chain_proof}
			D_\alpha(\rho_{AB}||\omega_A\otimes\sigma_B)&=\frac{1}{\alpha-1}\log\Tr\left[\left(\rho_{AB}\left(\omega_A\otimes\sigma_B\right)^\frac{1-\alpha}{\alpha}\right)^\alpha\right]\notag\\
			&=\frac{1}{\alpha-1}\log\Tr\left[\left(\rho_{AB}\rho_{A}^{-\frac{1}{2}}\rho_{A}^{\frac{1}{2}}\omega_A^\frac{1-\alpha}{\alpha}\rho_{A}^{\frac{1}{2}}\rho_{A}^{-\frac{1}{2}}\otimes\sigma_B^{\frac{1-\alpha}{\alpha}}\right)^\alpha\right]\notag\\
			&=\frac{1}{\alpha-1}\log\Tr\left[\left(\rho_{B|A}\eta^{\frac{1}{\alpha}}\otimes\sigma_B^{\frac{1-\alpha}{\alpha}}\right)^\alpha\right]+\frac{1}{\alpha-1}\log\Tr\left[\left(\rho_A^{\frac{1}{2}}\omega_A^{\frac{1-\alpha}{\alpha}}\rho_A^{\frac{1}{2}}\right)^\alpha\right]\notag\\
			&=\frac{1}{\alpha-1}\log\Tr\left[\left(\nu_{AB}\left(\eta_A\otimes\sigma_B\right)^{\frac{1-\alpha}{\alpha}}\right)^\alpha\right]+D_\alpha(\rho_A||\omega_A)\notag\\
			&=D_\alpha(\eta_{AB}||\eta_A\otimes\sigma_B)+D_\alpha(\rho_A||\omega_A),
		\end{align}
		where the first line is due to the cyclic property of trace, the third line follows from  Def.~\ref{def:cond_state} and substituting Eq.~(\ref{eq:eta}), and the fourth line follows from reshuffling terms and the definition of $\eta_{AB}$. To prove Eq.~\eqref{eq:minfo_1}, we take $\omega_A=\id_A$ in~\eqref{eq:C_chain}, and take an infimum from both sides. Thus we have:
		\begin{align}
			\label{eq:minfo_1proof}
			H_\alpha^\uparrow(A|B)_\rho&=-\inf_{\sigma_B\in\dop{=}(B)}D_\alpha(\rho_{AB}||\omega_A\otimes\sigma_B)\notag\\
			&=-D_\alpha(\rho_A||\id_A)-\inf_{\sigma_B\in\dop{=}(B)}D_\alpha(\eta_{AB}||\eta_A\otimes\sigma_B)\notag\\
			&=H_\alpha(A)_\rho-\IDA_\alpha(A;B)_\eta,
		\end{align}
		which completes the proofs.
	\end{proof}	
\end{lemma}

With this, we can obtain a bound involving a form of \Renyi\ conditional mutual information:
\begin{corollary}
	\label{cor:renyiCQMI}
	Let $\rho\in\dop{=}(ABC)$, and set $\eta_{ABC}\coloneqq\eta_A^{\frac{1}{2}}\rho_{BC|A}\eta_A^{\frac{1}{2}}$ where $\eta_A$ is the state in Eq.~\eqref{eq:eta}. Then, for any $\alpha\in [0,\infty)$:
	\begin{align}\label{eq:renyiCQMI}
		H_\alpha^\uparrow(A|BC)_\rho=H_\alpha^\uparrow(A|B)_\rho-\Idiff_\alpha(A;C|B)_\eta.
	\end{align}
	Moreover, let $\rho'\in\dop{=}(RA)$, and $\mathcal{N}:R\rightarrow BC$ be a channel such that $\rho_{ABC}=\mathcal{N}(\rho'_{RA})$. Then, Eq.~\eqref{eq:renyiCQMI} can be relaxed to
	\begin{align}
		\label{eq:renyiCQMI_relaxed}
		H_\alpha^\uparrow(A|B)_{\mathcal{N}(\rho')}-\sup_{\omega\in\dop{=}(RA)}\Idiff_\alpha(A;C|B)_{\mathcal{N}(\omega)}\leq H_\alpha^\uparrow(A|BC)_{\mathcal{N}(\rho')}\leq H_\alpha^\uparrow(A|B)_{\mathcal{N}(\rho')}-\inf_{\omega\in\dop{=}(RA)}\Idiff_\alpha(A;C|B)_{\mathcal{N}(\omega)}.
	\end{align}
	\begin{proof}
		We start by applying Lemma~\ref{lemma:C_chain} to the left side of the above equality:
		\begin{align}
			H_\alpha^\uparrow(A|BC)_\rho&=H_\alpha(A)_\rho-\IDA_\alpha(A;BC)_\eta\notag\\
			&=H_\alpha^\uparrow(A|B)_\rho+\IDA_\alpha(A;B)_\eta-\IDA_\alpha(A;BC)_\eta\notag\\
			&= H_\alpha^\uparrow(A|B)_\rho-\Idiff_\alpha(A;C|B)_\eta,
		\end{align}
		where the second line follows from applying Eq.~\eqref{eq:eta} and noting $\eta_{AB}=\Tr_C\eta_{ABC}=\eta_A^{\frac{1}{2}}\rho_{B|A}\eta_A^{\frac{1}{2}}$, and the last line follows from the definition of \Renyi\ conditional mutual information in Def.~\ref{def:sandwiched CQMI}.
		To prove the relaxed version let us first define $S_\rho=\left\{\sigma\in\dop{=}(ABC)\middle\vert\ \sigma_{BC|A}=\rho_{BC|A}\right \}$. Then, we have:
		\begin{align}
			H_\alpha^\uparrow(A|B)_\rho-\sup_{\sigma\in S_\rho}\Idiff_\alpha(A;C|B)_\sigma\leq H_\alpha^\uparrow(A|BC)_\rho\leq H_\alpha^\uparrow(A|B)_\rho-\inf_{\sigma\in S_\rho}\Idiff_\alpha(A;C|B)_\sigma,
		\end{align}
		where the inequalities hold by noting that by definition, $\eta_{BC|A}=\rho_{BC|A}$, therefore, $\eta_{ABC}\in S_\rho$. Thus, we only need to show that the optimizations in Eq.~\eqref{eq:renyiCQMI_relaxed} contain the set $S_\rho$. Let $\rho_{ABC}=\mathcal{N}(\omega_{RA})$ and define $\omega\in\dop{=}(RA)$ such that the following holds:
		\begin{align}
			\omega_{RA}=\eta_{A}^\frac{1}{2}\rho_A^{-\frac{1}{2}}\rho'_{RA}\rho_A^{-\frac{1}{2}}\eta_{A}^\frac{1}{2}.
		\end{align}  
		Then,
		\begin{align}
			\mathcal{N}(\omega_{RA})&=\eta_A^{\frac{1}{2}}\rho_A^{-\frac{1}{2}}\mathcal{N}(\rho'_{RA})\rho_A^{-\frac{1}{2}}\eta_A^{\frac{1}{2}}\notag\\
			&=\eta_{ABC}.
		\end{align}
		It is then clear that $\mathcal{N}(\omega_{RA})\in S_\rho$, which completes the proof.
	\end{proof}
\end{corollary}

As compared to Corollary~\ref{cor:mutual info chain}, the above chain rule in Corollary~\ref{cor:renyiCQMI} has the advantage that as $\alpha\to1$, it reproduces the equality $H(A|BC)_\rho = H(A|B)_\rho-I(A:C|B)_\rho$ for von Neumann entropies. Furthermore, it has the advantage that the channel can output both the $BC$ registers, not just one of them as in Corollary~\ref{cor:mutual info chain}. However, it also has some disadvantages due to $\Idiff_\alpha$ being potentially a less well-behaved quantity. (Note that we also cannot view Corollary~\ref{cor:mutual info chain} directly as a relaxation of Corollary~\ref{cor:renyiCQMI} either, because inspecting Def.~\ref{def:sandwiched CQMI} for $\Idiff_\alpha$ shows that the registers are ordered in a way that does not yield Corollary~\ref{cor:mutual info chain} as a simplified case.) Hence in the subsequent discussions, we state results based on both versions, leaving it up to individual applications which version is more useful.

\subsection{Incorporating the GEAT}
\label{subsec:chainwithleakage}

To incorporate the above chain rules into the entropy accumulation theorem, we need a version of that theorem with $H_\alpha^\uparrow$ entropies. To obtain such a version of entropy accumulation, we first present the following statement, which is a combination of Corollary~5.1 in~\cite{FF21}, with Theorem~3.1 and Lemma~3.5 in~\cite{arx_MFSR22}.
	\begin{fact}
		\label{fact:channel_chain}
		(Corollary~5.1 in~\cite{FF21}, with Theorem~3.1 and Lemma~3.5 in~\cite{arx_MFSR22}) Let $\rho\in\dop{=}(AA'R)$, $\sigma\in\Pos(AA'R)$, $\mathcal{M}$ be a channel $AA'R\rightarrow BB'R'$, and $\mathcal{N}=\mathcal{R}_{B}\circ\mathcal{M}$ with $\mathcal{R}_{B}$ being a replacer channel to identity (i.e., $\forall\omega\in\dop{=}(AB);\ \mathcal{R}_B(\omega)=\omega_A\otimes\id_B$). Then for any $\alpha\in(1,2)$, letting $\widehat{\alpha}\defvar1/(2-\alpha)$ we have
		\begin{align}
			\label{eq:channel_chain}
			D_\alpha(\mathcal{M}(\rho)||\mathcal{N}(\sigma))\leq D_\alpha(\rho||\sigma) + D_{\widehat{\alpha}}(\mathcal{M}||\mathcal{N}).
		\end{align}
Furthermore, if $\mathcal{N}$ satisfies the \term{no-signalling (NS) condition} that there exists a channel $\mathcal{E}:AA'\rightarrow BB'R'$ such that $\mathcal{N}=\mathcal{E}\circ\Tr_{R}$, then we have
		\begin{align}
			\label{eq:channel_chain_reduced}
			D_\alpha(\mathcal{M}(\rho)||\mathcal{N}(\sigma))\leq D_\alpha(\rho_{AA'}||\sigma_{AA'}) + D_{\widehat{\alpha}}(\mathcal{M}||\mathcal{N}).
		\end{align}
	\end{fact}
	We can now state a simple modification of Lemma~3.6 in~\cite{arx_MFSR22}, where we have $H_\alpha^\uparrow$ instead of $H_\alpha$. 
\begin{lemma}\label{lemma:GEAT}
	Let $\rho\in\dop{=}(SRE)$, $\mathcal{M}$ be a channel $RE\rightarrow S'R'E'$, and $\widetilde{E}$ be a purifying register for $RE$. Then for any $\alpha\in (1,2)$, letting $\widehat{\alpha}\defvar1/(2-\alpha)$ we have
	\begin{align}
		\label{eq:GEAT}
		H_\alpha^\uparrow(SS'|E')_{\mathcal{M}(\rho)}\geq H_\alpha^\uparrow(S|RE)_\rho+\inf_{\omega\in\dop{=}(RE\widetilde{E})}H_{\widehat{\alpha}}(S'|E'\widetilde{E})_{\Tr_{R'}\circ\mathcal{M}(\omega)}.
	\end{align}
Furthermore, if there exists a channel $\mathcal{R}:E\rightarrow E'$ such that $\Tr_{S'R'}\circ\mathcal{M}=\mathcal{R}\circ\Tr_{R}$ (i.e., the channel $\mathcal{M}$ is non-signalling from $R$ to $E'$), then we have:
	\begin{align}
		\label{eq:GEAT_NS}
		H_\alpha^\uparrow(SS'|E')_{\mathcal{M}(\rho)}\geq H_\alpha^\uparrow(S|E)_\rho+\inf_{\omega\in\dop{=}(RE\widetilde{E})}H_{\widehat{\alpha}}(S'|E'\widetilde{E})_{\Tr_{R'}\circ\mathcal{M}(\omega)}.
	\end{align}
	\begin{proof}
		Let $\tilde{\mathcal{M}}=\idmap_S\otimes\Tr_{R'}\circ\mathcal{M}$ and $\tilde{\mathcal{N}}=\mathcal{R}_{S'}\circ\tilde{\mathcal{M}}$ be channels $SRE\rightarrow SS'E'$, with $\mathcal{R}_{S'}$ being a replacer channel, then for any $\rho\in\dop{=}(SRE)$, define $\sigma\in\Pos(SRE)$ such that
		\begin{align}
			D_\alpha(\rho_{SRE}||\sigma_{SRE})=D_\alpha(\rho_{SRE}||\id_S\otimes\sigma^*_{RE})=\inf_{\sigma\in\dop{=}(RE)}D_\alpha(\rho_{SRE}||\id_S\otimes\sigma_{RE}).
		\end{align}
		Then, by definition it is clear that
		\begin{align}
			\label{eq:entropies}
			&D_\alpha\left(\tilde{\mathcal{M}}(\rho_{SRE})\middle\vert\middle\vert\tilde{\mathcal{N}}(\sigma_{SRE})\right)\geq\inf_{\omega\in\dop{=}(SE')}D_\alpha\left(\tilde{\mathcal{M}}(\rho_{SRE})\middle\vert\middle\vert\id_{S'}\otimes\omega_{SE'}\right)=-H_\alpha^\uparrow(SS'|E')_{\mathcal{M}(\rho)}
			\notag\\
			&D_{\widehat{\alpha}}(\tilde{\mathcal{M}}||\tilde{\mathcal{N}})=\sup_{\omega\in\dop{=}(RE\widetilde{E})}-H_{\widehat{\alpha}}(S'|E'\widetilde{E})_{\Tr_{R'}\circ\mathcal{M}(\omega)}\notag\\
			&D_\alpha(\rho_{SRE}||\sigma_{SRE})=-H_\alpha^\uparrow(S|RE)_\rho,
		\end{align} 
		where $\widetilde{E}$ is isomorphic to $RE$.
		Combining Eq.~\eqref{eq:channel_chain} with~\eqref{eq:entropies}, results in the following:
		\begin{align}
			H^\uparrow_\alpha(SS'|E')_{\mathcal{M}(\rho)}\geq H_\alpha^\uparrow(S|RE)_\rho +\inf_{\omega\in\dop{=}(RE\widetilde{E})}H_{\widehat{\alpha}}(S'|E'\widetilde{E})_{\Tr_{R'}\circ\mathcal{M}(\omega)},
		\end{align}
		where we used the fact that $H_\alpha^\uparrow\geq H_\alpha$. 
		
		To prove Eq.~\eqref{eq:GEAT_NS}, we use the same channels $\tilde{\mathcal{M}},\tilde{\mathcal{N}}$, and for any $\rho\in\dop{=}(SRE)$ we define $\sigma\in\Pos(SRE)$ to be any extension of $\sigma_{SE}$ such that
		\begin{align}\label{eq:entropies_NS}
				D_\alpha(\rho_{SE}||\sigma_{SE})=D_\alpha(\rho_{SE}||\id_S\otimes\sigma^*_{E})=\inf_{\sigma\in\dop{=}(E)}D_\alpha(\rho_{SE}||\id_S\otimes\sigma_{E}).
		\end{align}
		With this choice of $\sigma_{SRE}$, the first two lines in Eq.~\eqref{eq:entropies} remain unchanged. Then, furthermore, to employ Eq.~\eqref{eq:GEAT_NS}, we need to show that there exists a channel $\mathcal{E}:SE\rightarrow SS'E'$ such that $\tilde{\mathcal{N}}=\mathcal{E}\circ\Tr_{R}$. For any input state $\omega_{SE}$ we define $\mathcal{E}$ as follows:
		\begin{align}
			\mathcal{E}(\omega_{SE})&\coloneqq \id_{S'}\otimes\idmap_S\otimes\mathcal{R}(\omega_{SE})\notag\\
			&=\id_{S'}\otimes\idmap_S\otimes\mathcal{R}\circ\Tr_{R}(\omega_{SRE})\notag\\
			&=\id_{S'}\otimes\Tr_{S'}\circ\tilde{\mathcal{M}}(\omega_{SRE})\notag\\
			&=\tilde{\mathcal{N}}(\omega_{SRE}),
		\end{align}
		where the third line follows from non-signalling assumption on $\mathcal{M}$. Applying the chain rule in Eq.~\eqref{eq:channel_chain_reduced} using the first two terms in Eq.~\eqref{eq:entropies}, and~\eqref{eq:entropies_NS} we have:
		\begin{align}
			H_\alpha^\uparrow(SS'|E')_{\mathcal{M}(\rho)}\geq H_\alpha^\uparrow(S|E)_\rho+\inf_{\omega\in\dop{=}(RE\widetilde{E})}H_{\widehat{\alpha}}(S'|E'\widetilde{E})_{\Tr_{R'}\circ\mathcal{M}(\omega)},
		\end{align}
		where we used $H_\alpha^\uparrow\geq H_\alpha$.
	\end{proof}
\end{lemma}

\newcommand{\iniR}{R}
\newcommand{\iniE}{E}
\newcommand{\midR}{\widehat{R}}
\newcommand{\midE}{\widehat{E}}

With this, we can state a modification of entropy accumulation, which allows for handling some ``leakage processes''. 
Qualitatively, the idea is that before applying a channel $\mathcal{N}$ that satisfies the NS condition, we first apply a ``leakage channel'' $\LEAKchann$ that sends a leakage register $L$ from the memory register to the side-information registers. Using the chain rules from Sec.~\ref{sec:HandIbounds}, we can compensate for this by subtracting suitable mutual-information quantities involving the register $L$.
\begin{theorem}\label{thrm:GEATleakage}
	Let $\rho\in\dop{=}(S \iniR \iniE)$, let $\LEAKchann:\iniR\iniE\rightarrow L \midR \midE$ and $\mathcal{N}:L \midR \midE\rightarrow S'R'E'$ be channels, such that there exists a channel $\mathcal{R}:L \midE\rightarrow E'$, where $\Tr_{S'R'}\circ\mathcal{N}=\mathcal{R}\circ\Tr_{\midR}$ (i.e., $\mathcal{N}$ is non-signalling from $\midR$ to $E'$). Let $\widetilde{E}$ be a purifying register for $L \midR \midE$.
	Then, for any $\alpha\in(1,2)$, letting $\widehat{\alpha}\defvar1/(2-\alpha)$ we have:
	\begin{align}\label{eq:GEATleakage diff}
		H_\alpha^\uparrow(SS'|E')_{\mathcal{N}\circ\LEAKchann(\rho)}\geq H_\alpha^\uparrow(S|\iniE)_\rho-\sup_{\omega\in\dop{=}(S\iniR\iniE)}\Idiff_\alpha(S;L|\midE)_{\LEAKchann(\omega)}+\inf_{\omega\in\dop{=}(L\midR\midE\widetilde{E})}
		H_{\widehat{\alpha}}(S'|E'\widetilde{E})_{\mathcal{N}(\omega)}.
	\end{align}
	Furthermore, if the registers $\iniE,\midE$ are isomorphic and we restrict $\LEAKchann$ to act non-trivially only on register $\iniR$  (formally: there exists a channel $\widetilde{\LEAKchann}:\iniR\rightarrow L\midR$ such that $\LEAKchann=\widetilde{\LEAKchann}\otimes\idmap_{\iniE\rightarrow E}$), then we have:
	\begin{align}\label{eq:GEATleakage}
		H_\alpha^\uparrow(SS'|E')_{\mathcal{N}\circ\LEAKchann(\rho)} &\geq H_\alpha^\uparrow(S|\iniE)_\rho-\sup_{\omega\in\dop{=}(S\iniR E)}\IDA_\alpha(S\midE;L)_{\LEAKchann(\omega)}+\inf_{\omega\in\dop{=}(L\midR\midE\widetilde{E})}
		H_{\widehat{\alpha}}(S'|E'\widetilde{E})_{\mathcal{N}(\omega)} \nonumber\\
		&= H_\alpha^\uparrow(S|\iniE)_\rho-\sup_{\omega\in\dop{=}(\stabR\iniR)}\IDA_\alpha(\stabR;L)_{\widetilde{\LEAKchann}(\omega)}+
		\inf_{\omega\in\dop{=}(L\midR\midE\widetilde{E})}
		H_{\widehat{\alpha}}(S'|E'\widetilde{E})_{\mathcal{N}(\omega)},
	\end{align}
	where $\stabR$ is a purifying register for $\iniR$.
\end{theorem}
\begin{proof}
	Since $\mathcal{N}$ is non-signalling from $\midR$ to $E'$, by applying Lemma~\ref{lemma:GEAT}, we have
	\begin{align}
		H_\alpha^\uparrow(SS'|E')_{\mathcal{N}\circ\LEAKchann(\rho)}&\geq H_\alpha^\uparrow(S|L\midE)_{\LEAKchann(\rho)}+
		\inf_{\omega\in\dop{=}(L\midR\midE\widetilde{E})}
		H_{\widehat{\alpha}}(S'|E'\widetilde{E})_{\mathcal{N}(\omega)}\notag\\
		&\geq H_\alpha^\uparrow(S|\midE)_{\LEAKchann(\rho)}-\sup_{\omega\in\dop{=}(S\iniR\iniE)}\Idiff_\alpha(S;L|\midE)_{\LEAKchann(\omega)}+
		\inf_{\omega\in\dop{=}(L\midR\midE\widetilde{E})}
		H_{\widehat{\alpha}}(S'|E'\widetilde{E})_{\mathcal{N}(\omega)}\notag\\
		&\geq H_\alpha^\uparrow(S|\iniE)_{\rho}-\sup_{\omega\in\dop{=}(S\iniR\iniE)}\Idiff_\alpha(S;L|\midE)_{\LEAKchann(\omega)}+
		\inf_{\omega\in\dop{=}(L\midR\midE\widetilde{E})} 
		H_{\widehat{\alpha}}(S'|E'\widetilde{E})_{\mathcal{N}(\omega)},
	\end{align}
	where the first line follows from Eq.~\eqref{eq:GEAT_NS}, the second line is an application of Corollary~\ref{cor:renyiCQMI}, and the last line holds by data-processing (Fact~\ref{fact:DPI}). To prove Eq.~\eqref{eq:GEATleakage}, we use Corollary~\ref{cor:mutual info chain}:
	\begin{align}
		H_\alpha^\uparrow(SS'|E')_{\mathcal{N}\circ\LEAKchann(\rho)}&\geq H_\alpha^\uparrow(S|L\midE)_{\LEAKchann(\rho)}+
		\inf_{\omega\in\dop{=}(L\midR\midE\widetilde{E})}
		H_{\widehat{\alpha}}(S'|E'\widetilde{E})_{\mathcal{N}(\omega)}\notag\\
		&\geq 
		H_\alpha^\uparrow(S|\iniE)_{\rho}-\sup_{\omega\in\dop{=}(S\iniR\iniE)}\IDA_\alpha(S\iniE;L)_{\LEAKchann(\omega)}+
		\inf_{\omega\in\dop{=}(L\midR\midE\widetilde{E})}
		H_{\widehat{\alpha}}(S'|E'\widetilde{E})_{\mathcal{N}(\omega)},
	\end{align}
	where the first line holds by Eq.~\eqref{eq:GEAT_NS}, and the second line follows from Eq.~\eqref{eq:mutual info chain_relaxed} together with the fact that $\LEAKchann$ essentially acts as identity on $\iniE$.
\end{proof}

As another possibility for applications, we now also state a slight variant where rather than having a separate leakage register, the channel $\LEAKchann$ simply modifies the ``memory register'' $\iniR$, and we do not impose the NS condition on the subsequent $\mathcal{N}$ channel. If $\LEAKchann$ modifies the memory register in such a way that it has bounded mutual information with other registers (informally, it always ``disrupts'' any information stored on the memory register), this would also yield nontrivial bounds, despite the lack of NS condition. However, we do not currently have a concrete application in mind for this version.
\begin{corollary}
	\label{cor:GEATsignalling}
	Let $\rho\in\dop{=}(S\iniR\iniE)$, let $\LEAKchann: \iniR\iniE\rightarrow \midR\midE$ and $\mathcal{N}: \midR\midE\rightarrow S'R'E'$ be channels. Let $\widetilde{E}$ be a purifying register for $\midR\midE$. Then for any $\alpha\in(1,2)$, letting $\widehat{\alpha}\defvar1/(2-\alpha)$ we have:
		\begin{align}
		\label{eq:GEATsignalling}
		H_\alpha^\uparrow(SS'|E')_{\mathcal{N}\circ\LEAKchann(\rho)}\geq H_\alpha^\uparrow(S|\iniE)_\rho-\sup_{\omega\in\dop{=}(\iniR SE)}\Idiff_\alpha(S;\midR|\midE)_{\LEAKchann(\omega)}+\inf_{\omega\in\dop{=}(\midR\midE\widetilde{E})}H_{\widehat{\alpha}}(S'|E'\widetilde{E})_{\mathcal{N}(\omega)}.
	\end{align}
	Furthermore, if the registers $\iniE,\midE$ are isomorphic and we restrict $\LEAKchann$ to act non-trivially only on register $\iniR$  (formally: there exists a channel $\widetilde{\LEAKchann}:\iniR\rightarrow \midR$ such that $\LEAKchann=\widetilde{\LEAKchann}\otimes\idmap_{\iniE\rightarrow \midE}$), then we have:
	\begin{align}
		\label{eq:GEATsignallingreduced}
		H_\alpha^\uparrow(SS'|E')_{\mathcal{N}\circ\LEAKchann(\rho)}&\geq H_\alpha^\uparrow(S|\iniE)_\rho-\sup_{\omega\in\dop{=}(\iniR SE)}\IDA_\alpha(S\iniE;\midR)_{\LEAKchann(\omega)}+\inf_{\omega\in\dop{=}(\midR\midE\widetilde{E})}H_{\widehat{\alpha}}(S'|E'\widetilde{E})_{\mathcal{N}(\omega)} \nonumber\\
		&= H_\alpha^\uparrow(S|\iniE)_\rho - \sup_{\omega\in\dop{=}(\stabR\iniR)}\IDA_\alpha(\stabR;\midR)_{\widetilde{\LEAKchann}(\omega)} + \inf_{\omega\in\dop{=}(\midR\midE\widetilde{E})}H_{\widehat{\alpha}}(S'|E'\widetilde{E})_{\mathcal{N}(\omega)}.
	\end{align}
	where $\stabR$ is a purifying register for $\iniR$.
	\begin{proof}
		The proof is essentially the same as Theorem~\ref{thrm:GEATleakage}, by replacing the $L$ register with $\midR$, and noting that since we no longer have non-signalling condition we must use Eq.~\eqref{eq:GEAT} instead.
	\end{proof}
\end{corollary}

The standard GEAT relates the {\Renyi} conditional entropy of a state produced by a sequence of channels $\mathcal{M}_j$ (satisfying the NS condition) to the {\Renyi} conditional entropy produced by each individual channel, by iterating the chain rule in Fact~\ref{fact:channel_chain}. Our above results allow us to weaken the NS condition, by introducing ``leakage channels'' $\LEAKchann_j$ of the form considered in those results, although a little additional care is needed to account for the versions of the GEAT with ``testing''. For ease of future applications, we present formal theorem statements regarding this in Appendix~\ref{app:great}, for readers already familiar with the GEAT (and a strengthened version recently developed in~\cite{arx_AHT24}). 

\section{Example applications}
\label{sec:examples}

\newcommand{\dleak}{\delta_{\mathrm{leak}}}

Qualitatively, our results in the preceding section (and Appendix~\ref{app:great}) imply that if we analyze a sequence of $n$ channels in a GEAT-based framework, in which for each round a channel $\LEAKchann_j$ occurs that leaks some side-information to an adversary, then we can compensate for these leakage channels $\LEAKchann_j$ very simply by just subtracting the corresponding mutual-information values (summed over the $n$ channels $\LEAKchann_j$).
For simplicity we focus on the bound resulting from~\eqref{eq:GEATleakage}, since in that case, for each leakage channel the quantity of interest would be (omitting the $j$ subscript for brevity, and recalling that the $\widetilde{\LEAKchann}$ channel is the ``part'' of $\LEAKchann$ that only acts on $\iniR$):
\begin{align}\label{eq:simplebnd}
\sup_{\omega\in\dop{=}(\stabR \iniR)}\IDA_\alpha(\stabR ;L)_{\widetilde{\LEAKchann}(\omega)}.
\end{align} 
In other words, we simply need to upper bound the {\Renyi} mutual information between the channel output and some arbitrary purification of its input. Given a bound on the above quantity\footnote{The $\alpha$ value to use here would be determined by the $\alpha$ value used in the GEAT-based security proof. This is usually optimized numerically to maximize the keyrates, and so we mainly do not focus on specific values for it here; we simply note that its optimal value approaches $1$ as $n$ increases, as described in e.g.~\cite{DF19,arx_MFSR22,arx_AHT24}. Still, in Fig.~\ref{fig:infobound} later we show some example calculations for specific $\alpha$ values that may be of interest.} (summed over the rounds), we can very simply compensate for the leakage channels by just subtracting that value from the entropy bounds given by the GEAT without leakage.

We now outline some device-imperfection scenarios in which such a model can be considered, before explaining possible methods to explicitly bound the above quantity. 

\subsection{Types of device imperfections}

The most straightforward way to apply this model would be if there is an actual physical system $L_j$ that leaks to the adversary Eve in each round, and we can constrain its correlations to the secret data. For instance, this could be a photonic system that we constrain by requiring it to have a large vacuum component, as considered in~\cite{KZMW01,PCS+18} for device-dependent protocols and in~\cite{arx_Tan23} for device-independent protocols. In particular, this resolves an issue regarding a ``restricted adaptiveness'' condition imposed on Eve in~\cite{arx_Tan23} --- under the model we consider here, Eve is free to modify her attack in every round based on the leakage registers gathered from previous rounds. Hence our result allows us to handle such leakage mechanisms even against the most general adaptive attacks.

However, our method can also apply to less obviously related forms of device imperfection, such as source correlation. To sketch out the ideas, a typical prepare-and-measure protocol consists of Alice generating some secret classical value $S_j$ in each round, and then using some source device to prepare a quantum state (to be sent to Bob) that depends only on the value of $S_j$. Such a source can be handled using the GEAT in a PM protocol, as it satisfies the NS condition. The notion of source correlation refers to the possibility that a realistic source device may be imperfect, in that the state it generates might depend not only on the $S_j$ value in that round, but also the values used in previous rounds (or even more generally, it might even be entangled with the quantum states it prepared in previous rounds). 

While it may not be immediately apparent that our results can be used to study such a scenario, we now argue that in fact they can be, for some classes of source correlations. 
Specifically, our results can be applied whenever the source correlation can be modelled as follows (again leaving out the $j$ subscripts for brevity): after each round, the source retains a memory register $\iniR$ for the next round. Then each time Alice uses it to prepare a state, the source first prepares the ideal specified state (on some register, say, $A'$), then applies a channel of the form $\widetilde{\LEAKchann}: \iniR \to L$, and finally applies some channel $A'L \to A'$ to modify the state it prepared on $A'$.\footnote{Strictly speaking, for full accuracy in this model, various channels should have additional output registers to allow memory to be transferred onwards to future rounds. However, these are not pertinent to our discussion here and so we omit them for brevity.}

A concrete example of a simple scenario we can capture with this would be for instance if the source behaves as follows in each round:
\begin{itemize}
\item With some probability $1-\dleak$ the source perfectly prepares the ideal state on $A'$. 
\item Otherwise (with probability $\dleak$), it modifies the prepared state on $A'$ via some arbitrary quantum channel $A'\iniR\to A'$.
\end{itemize}
This can be captured within our model as follows. Let $L$ be a quantum register isomorphic to $\iniR$ except with one additional dimension given by a basis state we denote as $\ket{\perp}_L$. We describe the source device behaviour as:
\begin{itemize}
\item It first applies a channel $\widetilde{\LEAKchann}:\iniR \to L$, which with probability $1-\dleak$ completely ignores the input and outputs $\pure{\perp}_L$, and otherwise embeds the input state on $\iniR$ into the output register $L$.
\item It then perfectly prepares the ideal state on $A'$. 
\item It then performs a two-outcome projective measurement on $L$, onto the subspace spanned by $\ket{\perp}_L$ and the orthogonal subspace. Conditioned on the outcome, it either simply leaves the state on $A'$ untouched, or it modifies it with some arbitrary quantum channel $A'L \to A'$.
\end{itemize}
Note that this will indeed reproduce the same states as the source device model described above, and is compatible with our approach of interleaving leakage channels between channels satisfying the NS conditions.\footnote{A technical point is that in the context of our full Theorem~\ref{thrm:GREATleakage}--\ref{th:GREAT} statements in Appendix~\ref{app:great}, the $\widetilde{\LEAKchann}$ channel in the above description would be referring to the leakage channel $\LEAKchann_{j-1}$ of the \emph{preceding} round in Theorems~\ref{thrm:GREATleakage}--\ref{th:GREAT}, in order to model this process correctly within that framework.} Still, we remark that this is just a starting example --- we could consider more elaborate models to allow more general forms of source correlations. For instance, rather than the essentially ``classical'' form of imperfect behaviour considered above, one could impose only the constraint that the generated states are close in fidelity to the ideal states (along the lines of~\cite{LPK23,MCA23,CNLT24}), and aim to construct a corresponding model in our framework in which the ``leaked'' mutual information can be bounded. We aim to study such scenarios in future work.

\subsection{Bounding the mutual information}

\newcommand{\CvsQ}{\zeta_\mathrm{cq}}

We now turn to the question of bounding~\eqref{eq:simplebnd}. If we consider the simple example described above, a straightforward bound can be obtained as follows. (This analysis also applies in the context of photonic leakage as described earlier, if it follows the ``classical-probabilistic leakage'' model described in~\cite{arx_Tan23}; refer to that work for details of that model.)

First note that in the described model, we can write $\widetilde{\LEAKchann}$ as $(1-\dleak) \widetilde{\LEAKchann}^{\perp} + \dleak \widetilde{\LEAKchann}^\mathrm{leak}$ for some channels $\widetilde{\LEAKchann}^{\perp}$ and $\widetilde{\LEAKchann}^\mathrm{leak}$, where the former just outputs a fixed state $\pure{\perp}_L$.
Therefore, for any possible output state $\nu_{RL} \defvar \widetilde{\LEAKchann}(\omega_{\stabR \iniR})$ of the channel (with a purifying system $\stabR$), by writing $\nu'_{\stabR L} \defvar \widetilde{\LEAKchann}^\mathrm{leak}(\omega_{\stabR \iniR})$, we have
\begin{align}
\nu_{\stabR L} = (1-\dleak)\nu_{\stabR} \otimes \pure{\perp}_L + \dleak \nu'_{\stabR L}, \text{ where $\bra{\perp}\nu'_L\ket{\perp} = 0$ and $\nu'_{\stabR} = \nu_{\stabR}$},
\end{align}
where the $\nu'_{\stabR} = \nu_{\stabR}$ property follows from the fact that all these channels did not act on $\stabR$. With this we have a simple bound:
\begin{align}\label{eq:simplebound}
\IDA_\alpha(\stabR ;L)_\nu &\leq I_\alpha(\stabR :L)_\nu \nonumber\\ 
&= D_\alpha( (1-\dleak)\nu_{\stabR} \otimes \pure{\perp}_L + \dleak \nu'_{\stabR L} ||\nu_{\stabR}\otimes ((1-\dleak)\pure{\perp}_L + \dleak \nu'_{L})) \nonumber\\
&= \frac{1}{\alpha-1} \log \left( (1-\dleak) 2^{(\alpha-1)(0)} + \dleak 2^{(\alpha-1)D_\alpha(\nu'_{\stabR L} || \nu_{\stabR} \otimes \nu'_L)} \right) \nonumber\\
&= \frac{1}{\alpha-1} \log \left( (1-\dleak) + \dleak 2^{(\alpha-1)D_\alpha(\nu'_{\stabR L} || \nu'_{\stabR} \otimes \nu'_L)} \right) \nonumber\\
&\leq \frac{1}{\alpha-1} \log \left( (1-\dleak) + \dleak 2^{(\alpha-1) \CvsQ \log \min\{\dim(\stabR ),\dim(L)\}} \right) \quad \forall \alpha \in (1,3/2)
,
\end{align}  
where $\CvsQ\defvar1$ if either $\stabR$ or $L$ is classical and $\CvsQ\defvar2$ otherwise. In the above, the third line holds by~\eqref{eq:classmixD} (and the fact that all the terms with $\nu'$ have disjoint support from the terms with $\pure{\perp}_L$), the fourth line uses the critical property $\nu'_{\stabR} = \nu_{\stabR}$, and the last line is from~\cite[proof of Corollary III.5]{DF19}. (In our specific example we always have $\dim(\stabR )<\dim(L)$ so the minimum in the last line is quite redundant, but we leave it in this form for other potential applications.) While the dimension dependence might appear inconvenient, we note that if for instance the memory $\iniR$ only stores memory about a bounded number of previous rounds, then its dimension (and thus also that of the purifying system $\stabR$) can be bounded. Alternatively, as observed in works such as~\cite{Win16,arx_Tan23}, it should be possible to replace it with weaker conditions (such as an energy bound); we leave this for future work. 

While the above bound might look like a slightly complicated formula, it is easy to compute, and straightforwardly verified to go to zero as $\dleak\to0$, as we would expect. Slightly less straightforwardly, it is also decreasing as $\alpha\to1$, and in that limit it simply converges to 
\begin{align}\label{eq:vNbound}
\dleak \CvsQ \log \min\{\dim(\stabR ),\dim(L)\},
\end{align}
which is the upper bound we would get if we had considered the von Neumann mutual information $I(\stabR :L)_\nu$.\footnote{Interestingly, this bound is linear in $\dleak$, in contrast to previous bounds computed in~\cite[Sec.~4.1]{arx_Tan23} (which applied to the ``classical-probabilistic leakage'' scenario), which reduced to $\sup_{\omega}H(L)_{\widetilde{\LEAKchann}(\omega)}$ asymptotically. The latter is significantly worse, in that it is of order $O(\sqrt{\dleak})$ at small $\dleak$ and hence grows faster than any linear function there, i.e.~it is not Lipschitz continuous. It appears that by working with mutual information rather than entropies, we were able to exploit more structure in this particular leakage model to obtain better bounds. } This is consistent with the fact that in GEAT-based security proofs we usually take $\alpha\to1$ to obtain convergence to the asymptotic rates.

\begin{remark}\label{remark:CQ}
In the above bounds, we defined $\CvsQ$ in terms of whether $\stabR$ or $L$ is classical, but we highlight that whenever the memory register $\iniR$ is classical, we can also take $\stabR$ to be classical without loss of generality. The argument is fairly intuitive, though somewhat tedious. First observe that in the optimization~\eqref{eq:simplebnd}, a condition that $\iniR$ is classical can be equivalently reformulated by stating that the channel $\widetilde{\LEAKchann}$ always begins by applying a pinching channel $\mathcal{P}_{\iniR}$ in the classical basis of $\iniR$, and taking the supremum over all (fully quantum) $\omega_{\stabR \iniR}$ rather than only classical-quantum $\omega_{\stabR \iniR}$.\footnote{Strictly speaking, this formulation is needed anyway to rigorously formalize the constraint that $\iniR$ is classical, since our Theorem~\ref{thrm:GEATleakage} statement technically requires the optimization to be taken over arbitrary (fully quantum) input states.} We now take any arbitrary $\omega_{\stabR \iniR}$ attaining the supremum (this always exists, by compactness and continuity), and perform the following sequence of transformations to find another optimizer $\widetilde{\omega}_{\stabR \iniR}$ that is fully classical. First note that $\mathcal{P}_{\iniR}[\omega_{\stabR \iniR}]$ also attains the supremum, since $\mathcal{P}_{\iniR} \circ \mathcal{P}_{\iniR} = \mathcal{P}_{\iniR}$ and so we can always pinch the state \emph{before} supplying it to $\widetilde{\LEAKchann}$. Now, letting $\omega'_{\iniR}$ denote the reduced state of $\mathcal{P}_{\iniR}[\omega_{\stabR \iniR}]$ on $\iniR$, it must have the form $\omega'_{\iniR} = \sum_j \lambda_j \pure{j}_{\iniR}$ where $\ket{j}_{\iniR}$ is the classical basis, due to the pinching. By a standard data-processing argument, \emph{any} purification $\pure{\omega'}_{\stabR \iniR}$ of $\omega'_{\iniR}$ attains the supremum as well; in particular we can choose $\ket{\omega'}_{\stabR \iniR} \defvar \sum_j \sqrt{\lambda_j} \ket{jj}_{\stabR \iniR}$. Finally, again recalling that we can always pinch the initial state without affecting the results, we conclude that $\widetilde{\omega}_{\stabR \iniR} \defvar \mathcal{P}_{\iniR}[\pure{\omega'}_{\stabR \iniR}] = \sum_j \lambda_j \pure{jj}_{\stabR \iniR}$ also attains the supremum, and it is fully classical as desired (here we exploited the fact that for a state of the form $\sum_j \sqrt{\lambda_j} \ket{jj}_{\stabR \iniR}$, pinching on $\iniR$ alone is sufficient to also make it classical on $\stabR$).
\end{remark}

\begin{figure}
	\centering
	\subfloat{
		\includegraphics[width=0.49\textwidth]{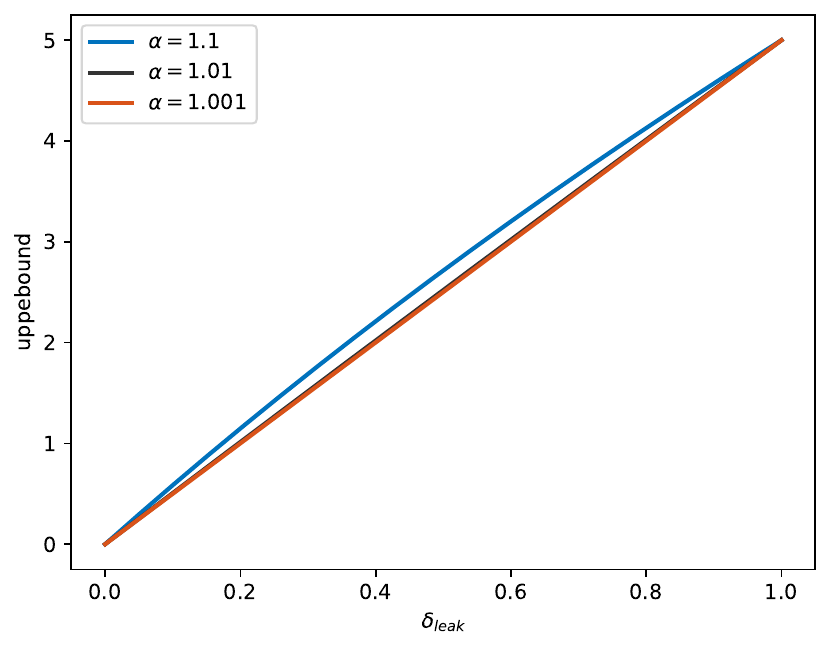}
	}
	\subfloat{
		\includegraphics[width=0.51\textwidth]{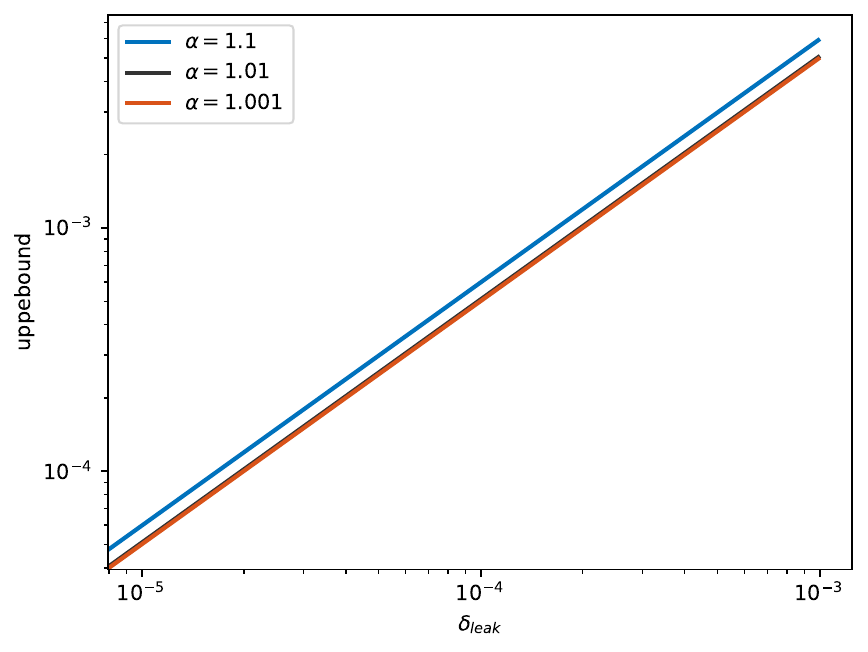}
	}
	\caption{Left: the upper bound on {\Renyi} mutual information in Eq~\eqref{eq:simplebound} versus leakage probability $\dleak$, for three different values of \Renyi\ parameter $\alpha\in\{1.1,1.01,1.001\}$, when the memory register is taken to be a classical register of $5$ bits (so $\dim(\stabR )=\dim(\iniR)=2^5$ and $\CvsQ=1$; see Remark~\ref{remark:CQ}). Right: the same plot focused on the small-$\dleak$ regime, and on log scales for the axes. We omit a plot of the von Neumann case (Eq.~\eqref{eq:vNbound}) because we found that it was visually indistinguishable from the $\alpha=1.001$ case; it can be seen from the plots that the $\alpha=1.01$ case is also already very close to this case. Note that existing GEAT security proofs usually choose $\alpha$ significantly smaller than $1.1$ at typical protocol sample sizes, and thus the $\alpha=1.1$ plots in the above figures serve as a conservative upper bound on how much the keyrates decrease when applying our techniques to account for leakage (under the model described in the main text).}
	\label{fig:infobound}
\end{figure}

To get a sense of how our bound in Eq.~\eqref{eq:simplebound} behaves, we plot it for some simple examples in Fig.~\ref{fig:infobound}. 
Those plots provide a complete description of how much the keyrates would decrease when applying our techniques to account for leakage under the above model, since our analysis states that the \emph{only} change to make compared to the case without leakage is to simply subtract the value in Eq.~\eqref{eq:simplebound} from the keyrates. While it is technically a function of the $\alpha$ value used in the security proof, note that the bound is nondecreasing in $\alpha$ and we found in~\cite{arx_KAG+24} that GEAT-based security proofs would usually take $\alpha<1.1$ in typical parameter regimes, so the $\alpha=1.1$ case in that figure provides a conservative upper bound on how much the keyrates would decrease in typical applications. 
We find that the bound should be reasonably tight for practical applications, and also converges to the von Neumann version quite quickly as $\alpha\to1$.

As briefly mentioned in Sec.~\ref{sec:Info_Acc}, an alternative to the above approach would be to first write $\IDA_\alpha(\stabR ;L)_\nu \leq I(\stabR :L)_\nu + O(\alpha-1)$ using the continuity bounds in~\cite{DFR20,DF19}, then bound the von Neumann mutual information $I(\stabR :L)_\nu$. However, this often has somewhat worse performance because for instance, at any \emph{fixed} $\alpha$ the resulting bound does not converge to zero even as $\dleak\to0$, due to the $O(\alpha-1)$ term. While it is true that one usually takes $\alpha\to1$ in the large-$n$ regime of GEAT-based security proofs, this difference can be significant at smaller $n$. 

Again, we highlight that the bounds we presented here are just a fairly simple starting demonstration --- since our results reduce the problem to studying an intuitive measure of correlation (the {\Renyi} mutual information) between the single-round leakage registers and the memory, there should be other ways to bound it. Another model worth studying is for instance if we just constrain $L$ to have a large vacuum component, as briefly noted above, and apply techniques similar to those in~\cite{arx_Tan23}.\footnote{Note however that if we formalize this task as e.g.~evaluating $\sup_{\nu}\IDA_\alpha(\stabR ;L)_{\nu} $ under \emph{only} the constraint of a dimension bound and a constraint $\bra{\phi}\nu\ket{\phi}_L \geq 1-\dleak$ where $\ket{\phi}_L$ denotes the vacuum state, then (focusing on the $\alpha=1$ asymptotic scenario) it is not hard to show that $\sup_{\nu}I(\stabR :L)_{\nu} \leq \sup_{\nu} 2 H(L)_{\nu}$ by taking purifications of $\nu_L$, and this bound is saturated whenever $\dim(\stabR)\geq\dim(L)$. (If we restrict to classical $L$, we instead get $\sup_{\nu}H(L)_{\nu}$.) 
As briefly mentioned above, this matches the asymptotic results for this model in~\cite{arx_Tan23}, which implies in particular that it unfortunately grows faster than any linear function of $\dleak$, unlike the bound~\eqref{eq:vNbound}. Still, there may perhaps be some room for refining the analysis by more closely studying the requirement that $\nu_{\stabR L}$ must always be some output state of a fixed channel. Also, this equality may not precisely generalize to the $\alpha>1$ case, though it should hold in some slightly modified form. We leave more detailed analysis for future work.} 
We leave a detailed exploration for future work.

\begin{remark}
A natural question regarding the above analysis is whether the dimension dependence can be removed. However, a fundamental obstacle that must be addressed in such contexts is the issue of the ``random full-leakage attack'' described in~\cite[Sec.~2.2]{arx_Tan23} (or~\cite[Appendix~E]{MD23}), which inherently renders a protocol completely insecure despite each round only ``behaving badly'' with probability $\dleak$. Any imperfection-robust security proof \emph{necessarily} has to impose enough assumptions to rule out that attack. In this case we chose to do so by assuming the dimensions of $L$ or $\stabR$ are bounded (so that~\eqref{eq:simplebound} is a nontrivial bound), but our general results highlight that in principle, any other reasonably plausible assumption that bounds the correlation between the single-round leakage register and other registers (specifically, as quantified by $\IDA_\alpha(\stabR ;L)$) would suffice. As mentioned above, another alternative to dimension bounds would be energy bounds, given suitable Hamiltonians; see~\cite{arx_Tan23,Win16}.
\end{remark}

\section{Conclusion}

Overall, in this work we have derived a variety of chain rules involving mutual information quantities, including both one-shot chain rules and {\Renyi} chain rules suitable for repeated application over a seequence of channels. In particular, the latter yield modifications of entropy accumulation that can accommodate some leakage, which are fairly straightforward to incorporate into security proofs. These results should be useful in rigorously quantifying the effects of device imperfections in protocol implementations. 

We also remark that the one-shot chain rules in Sec.~\ref{sec:Imax_chain} may potentially be relevant in the context of composable security~\cite{arx_PR21}. A full discussion of this topic is well beyond the scope of this work, but to informally summarize, the composable security of QKD protocols may encounter issues if additional information $L$ about the raw key is leaked after the protocol has completed~\cite{BCK13}. Our one-shot chain rules suggest that it might theoretically be possible to compensate for such additional leakage by shortening the output of the privacy amplification step in QKD~\cite{rennerthesis} by approximately $\ImaxDA[,\delta](SE;L)$ (or the potentially tighter quantity in Theorem~\ref{thrm:Hmin_Imax chain}), if that value can be bounded. In particular, this bound would likely be better than simply subtracting the log-dimension or smooth max-entropy of $L$, since $\ImaxDA[,\delta](SE;L)$ does at least reduce to zero in trivial scenarios where $L$ is in product with $SE$.
However, we leave a detailed analysis of whether such an approach truly yields composable security to future work. (For now we simply highlight that in particular, it likely does not work if the leakage ``directly depends on'' the \emph{final} key rather than the \emph{raw} key, e.g.~if the leakage is simply an exact copy of some bits in the final key.)

\section*{Acknowledgements}

We thank Mario Berta, Ian George, Srijita Kundu, Ashutosh Marwah, Renato Renner, Marco Tomamichel, and Mark Wilde for helpful feedback and discussions.
AA and EYZT conducted research at the Institute for Quantum Computing, at the University of Waterloo, which is supported by Innovation, Science, and Economic Development Canada. Support was also provided by NSERC under the Discovery Grants Program, Grant No. 341495.
TM acknowledges support from the ETH Z\"{u}rich Quantum Center, the SNSF QuantERA project (grant 20QT21\_187724), the AFOSR grant FA9550-19-1-0202, and an ETH Doc.Mobility Fellowship.

\appendix

\section{Incorporating entropy accumulation with testing}
\label{app:great}
In this appendix, we incorporate the chain rule we obtained in Sec.~\ref{sec:HandIbounds} into a full entropy accumulation statement with ``testing'' registers to estimate the accumulated entropy, to allow handling the leakage structure within that framework as well. To focus on the tightest existing bounds, we discuss only the ``fully {\Renyi}'' GEAT recently developed in~\cite{arx_AHT24} --- as shown in that work (Sec.~5.2), the resulting bounds can be easily relaxed to more closely correspond to earlier results based on von Neumann entropy and smooth min-entropy. 

For brevity, we only briefly state the required definitions, deferring detailed explanations to~\cite{arx_AHT24}. We also only state results based on our chain rule in Eq.~\eqref{eq:GEATleakage} rather than Eq.~\eqref{eq:GEATleakage diff}, as we believe that the {\Renyi} mutual information term resulting from the former will be easier to analyze in most circumstances; however, all the statements presented below would straightforwardly generalize to the latter version as well. (As previously discussed, the latter version may be a somewhat tighter bound in some situations, but the {\Renyi} CQMI term in it might not satisfy properties such as data-processing.)

\begin{definition}\label{def:freq}
(Frequency distributions) For a string $z_1^n\in\mathcal{Z}^n$ on some alphabet $\mathcal{Z}$, $\freq_{z_1^n}$ denotes the following probability distribution on $\mathcal{Z}$:
\begin{align}
\freq_{z_1^n}(z) \defvar \frac{\text{number of occurrences of $z$ in $z_1^n$}}{n} .
\end{align}
\end{definition}

\begin{definition}
	(Measure-and-prepare or read-and-prepare channels)
	A (projective) \term{measure-and-prepare channel} is a channel $\mathcal{E}: Q \to QQ'$ of the form
	\begin{align}
		\mathcal{E}[\rho_Q] = \sum_j (P_j \rho_Q P_j) \otimes \sigma_{Q'|j},
	\end{align}
	for some projective measurement $\{P_j\}$ on $Q$ and some normalized states $\sigma_{Q'|j}$.
	If $Q$ is classical and the measurement is a projective measurement in its classical basis, we shall refer to it as a \term{read-and-prepare channel}.
	Note that a read-and-prepare channel always simply extends the state ``without disturbing it'', i.e. tracing out $Q'$ results in the original state again.
\end{definition}
\begin{definition}\label{def:GEATTL}
$\{\EATchann_j\}_{j=1}^n$ is called a \term{sequence of GEAT-with-testing-and-leakage (GEATTL) channels} if each $\EATchann_j$ is a channel $R_{j-1}E_{j-1}\rightarrow S_jR_jE_j\CS_j\CP_j$ such that the output registers $\CS_j\CP_j$ are always classical (for any input states), and it can be written as $\EATchann_j=\mathcal{N}_j\circ\LEAKchann_j$, where $\LEAKchann_j: R_{j-1}\rightarrow \widehat{R}_{j-1}L_j$ and $\mathcal{N}_j: \widehat{R}_{j-1}L_jE_{j-1}\rightarrow S_jR_jE_j\CS_j\CP_j$, with $\mathcal{N}_j$ satisfying the following \term{non-signalling (NS) condition}:
\begin{align}
\exists \text{ a channel } \mathcal{R}_j: L_jE_{j-1}\rightarrow E_j\widehat{C}_j \text{ such that } \Tr_{S_jR_j\CS_j}\circ\mathcal{N}_j=\mathcal{R}_j\circ\Tr_{\widehat{R}_{j-1}}.
\end{align} 
We refer to the channels $\LEAKchann_j$ as the \term{leakage channels}. If a state $\rho \in \dop{=}(S_1^n \CS_1^n \CP_1^n E_n R_n)$ is of the form $\rho=\EATchann_n \circ \dots \circ \EATchann_1 [\omega^0]$
	(leaving some identity channels implicit) 
	for some initial state $\omega^0 \in \dop{=}(R_0 E_0)$, we say it is \term{generated by the sequence of GEATTL channels $\{\EATchann_j\}_{j=1}^n$}.
\end{definition}
Qualitatively, 
$S_j$ can be understood as (possibly quantum) registers that will be kept secret,  
$E_j$ represents some side-information that can be updated by each GEAT channel, and $R_j$ represents memory passed between the channels without being available as side-information except via ``leakage'' through the registers $L_j$. The classical registers $\CS_j$, $\CP_j$ respectively contain secret and public information that will be used to estimate the accumulated entropy. The NS condition is basically the statement that after the leakage channel is applied,  there is no further ``signalling'' of information from the memory register $\widehat{R}_{j-1}$ to the side-information registers $E_j\widehat{C}_j$, in the sense that if we consider the term $\mathcal{R}_j\circ\Tr_{\widehat{R}_{j-1}}$ on the right-hand-side, it is a channel that outputs the ``correct'' reduced state on the side-information $E_j\widehat{C}_j$ (namely, the state produced by $\Tr_{S_jR_j\CS_j}\circ\mathcal{N}_j$) despite first tracing out $\widehat{R}_{j-1}$.

\begin{definition}\label{def:QES}
	Let $\rho \in \dop{=}(\CS \CP Q Q')$ be a state where $\CS$ and $\CP$ are classical with alphabets $\alphCS$ and $\alphCP$ respectively. A \term{quantum estimation score-system (QES) on $\CS \CP$} is simply a function $f:\alphCS \times \alphCP \to \mathbb{R}$; equivalently, we may denote it as a real-valued tuple $\mbf{f}
	\in \mathbb{R}^{|\alphCS \times \alphCP|}$ where each term in the tuple specifies the value $f(\cS \cP)$. Given a QES $f$ and a value $\alpha\in(0,1)\cup (1,\infty)$, the \term{QES-entropy of order $\alpha$} for $\rho$ is defined as
	\begin{align}\label{eq:QESdefn}
		H^{f}_\alpha(Q\CS|\CP Q')_{\rho} &\defvar \frac{1}{1-\alpha} \log \left( \sum_{\cS \cP} \rho(\cS \cP)^\alpha \rho(\cP)^{1-\alpha} \, 2^{(1-\alpha) \left(-f(\cS \cP) - D_\alpha\left(\rho_{QQ'|\cS\cP} \middle\Vert \id_{Q} \otimes \rho_{Q'|\cP} \right)\right) } \right) \nonumber\\
		&=\frac{1}{1-\alpha} \log \left( \sum_{\cS \cP} \rho(\cS \cP) 2^{(1-\alpha) \left(-f(\cS \cP) - D_\alpha\left(\rho_{QQ' \land \cS\cP} \middle\Vert \id_{Q} \otimes \rho_{Q' \land \cP} \right)\right) } \right) \nonumber\\
		&= \frac{1}{1-\alpha} \log \left( \sum_{\cS \cP}  
		2^{-(1-\alpha)f(\cS \cP)} 
		\Tr \left[\left(\left(
		\rho_{Q' \land \cP}\right)^{\frac{1-\alpha}{2\alpha}}\rho_{QQ' \land \cS\cP}\left(
		\rho_{Q' \land \cP}\right)^{\frac{1-\alpha}{2\alpha}}\right)^\alpha\right]
		\right)
		,
	\end{align}
	where the sum is over all $\cS\cP$ values such that $\rho(\cS\cP)>0$, and we leave some tensor factors of identity implicit in the last expression.
\end{definition}

Qualitatively, the QES values $f(\cS \cP)$ in the above definition are ``scores'' that serve to estimate the entropy of the state; however, the concept is fairly elaborate and we defer to~\cite{arx_AHT24} for detailed discussions and intuitive explanations.
We now state the following theorem which is a modification of~\cite[Theorem~4.1]{arx_AHT24} in order to incorporate leakage structure.
\begin{theorem}[QES-entropy accumulation with leakage]\label{thrm:GREATleakage}
	Let $\rho$ be a state generated by a sequence of GEATTL channels $\{\EATchann_j\}_{j=1}^n$ (Definition~\ref{def:GEATTL}), so they have the form $\EATchann_j=\mathcal{N}_j\circ\LEAKchann_j$. For each j, suppose that for every value $\cS_1^{j-1}\cP_1^{j-1}$, we have a QES $f_{|\cS_1^{j-1} \cP_1^{j-1}}$ on registers $\CS_j\CP_j$. Define the following QES on $\CS_1^n\CP_1^n$:
	\begin{align}\label{fullQES}
		f_\mathrm{full}(\cS_1^n \cP_1^n) \defvar \sum_{j=1}^n f_{|\cS_1^{j-1} \cP_1^{j-1}}(\cS_j \cP_j).
	\end{align}
	Take any $\alpha \in (1,2)$
	and let $\widehat{\alpha}=1/(2-\alpha)$.
	Then we have
	\begin{align}\label{eq:chainQES}
		H^{f_\mathrm{full}}_{2-\frac{1}{\alpha}}(S_1^n \CS_1^n | \CP_1^n E_n)_\rho &\geq \sum_j \left(\min_{\cS_1^{j-1} \cP_1^{j-1}} \kappa_{\cS_1^{j-1} \cP_1^{j-1}}-\xi_j\right), \notag\\ 
		\text{where} \quad \kappa_{\cS_1^{j-1} \cP_1^{j-1}} &\defvar \inf_{\nu\in\Sigma^\mathcal{N}_j} H^{f_{|\cS_1^{j-1} \cP_1^{j-1}}}_{\widehat{\alpha}}(S_j \CS_j | \CP_j E_j \widetilde{E})_{\nu},\\
		\xi_j &\defvar\sup_{\nu\in\Sigma^\LEAKchann_j} \IDA_\alpha(\stabR ;L_j)_\nu,\notag
	\end{align}
	where $\Sigma_j^\mathcal{N}$ denotes the set of all states of the form $\mathcal{N}_j\left[\omega_{\widehat{R}_{j-1} L_j E_{j-1} \widetilde{E}}\right]$ for some initial state $\omega \in \dop{=}(\widehat{R}_{j-1} L_j E_{j-1} \widetilde{E})$, with $\widetilde{E}$ being a register of large enough dimension to serve as a purifying register for any of the $\widehat{R}_{j-1} L_j E_{j-1} $ registers, and $\Sigma_j^\LEAKchann$ denotes the set of all states of the form $\LEAKchann_j\left[\omega''_{R_{j-1}\stabR}\right]$ for some initial state $\omega'' \in \dop{=}(R_{j-1}\stabR )$, with $\stabR$ being a register of large enough dimension to serve as a purifying register for any of the $\widehat{R}_{j-1}$ registers.
	
	Consequently, if we instead define the following ``normalized'' QES on $\CS_1^n \CP_1^n$:
	\begin{align}\label{eq:fullQESnorm}
		\hat{f}_\mathrm{full}(\cS_1^n \cP_1^n) \defvar \sum_{j=1}^n \hat{f}_{|\cS_1^{j-1} \cP_1^{j-1}}(\cS_j \cP_j), \quad\text{where } \hat{f}_{|\cS_1^{j-1} \cP_1^{j-1}}(\cS_j \cP_j) \defvar f_{|\cS_1^{j-1} \cP_1^{j-1}}(\cS_j \cP_j) + \kappa_{\cS_1^{j-1} \cP_1^{j-1}} - \xi_j,
	\end{align}
	then
	\begin{align}\label{eq:chainQESnorm}
		H^{\hat{f}_\mathrm{full}}_{2-\frac{1}{\alpha}}(S_1^n \CS_1^n | \CP_1^n E_n)_\rho \geq 0.
	\end{align}
	\begin{proof}
		Let $M>0$ be any value such that $M - f_{|\cS_1^{j-1} \cP_1^{j-1}}(\cS_j \cP_j) > \frac{M}{2} > 0$ for all the $f_{|\cS_1^{j-1} \cP_1^{j-1}}(\cS_j \cP_j)$ in the theorem statement (for all $j$). 
		Define a read-and-prepare channel $\mathcal{D}_j:\CS_1^j \CP_1^j \to \CS_1^j \CP_1^j D_j$ for each $j$, the same way as the proof of~\cite[Theorem~4.1]{arx_AHT24}.
		Let $\mathcal{E}_j: \widehat{R}_{j-1} L_j E_{j-1}  \CS_1^{j-1} \CP_1^{j-1} \to D_j \copyCS_j S_j R_j E_j \CS_1^j \CP_1^j$ denote a channel that does the following:
		\begin{enumerate}
			\item Apply $\mathcal{N}_j \otimes \mathcal{P}_j$, where $\mathcal{P}_j$ is a pinching channel on $\CS_1^{j-1} \CP_1^{j-1}$ (in its classical basis). 
			\item Copy the classical register $\CS_j$ onto another classical register $\copyCS_j$.
			\item Generate a $D_j$ register by applying the above read-and-prepare channel $\mathcal{D}_j$ on $\CS_1^j \CP_1^j$.
		\end{enumerate}
		
		As argued in the~\cite[Theorem~4.1]{arx_AHT24} proof, these channels $\mathcal{E}_j$ inherit the NS condition of the $\mathcal{N}_j$ channels, in that they are non-signalling from $\widehat{R}_{j-1}\CS_1^{j-1}$ to $E_j\widehat{C}_j$. Also, the state $\mathcal{E}_n \circ \LEAKchann_n \circ \dots \mathcal{E}_1 \circ \LEAKchann_1 [\omega^0]$ is a valid extension of the state $\rho$ in the theorem statement, so we can use $\rho$ to refer to this state as well without ambiguity. By now applying  Theorem~\ref{thrm:GEATleakage} to this state (since the $\mathcal{E}_j$ channels satisfy the appropriate NS condition), we obtain:
		\begin{align}\label{eq:GEATbound_notest}
			H_\alpha^\uparrow(D_1^nS_1^n\copyCS_1^n|\CP_1^nE_n)_\rho \geq \sum_j \left(\inf_{\nu'\in\Sigma_j^\mathcal{E}}H_{\widehat{\alpha}}(D_jS_j\copyCS_j|\CP_1^jE_j\widetilde{E})_{\nu'}-\sup_{\nu'\in\Sigma_j^\LEAKchann}\IDA_\alpha(\stabR ;L_j)_{\nu'}\right)
		\end{align}
		where $\Sigma_j^\mathcal{E}$ is the set of all states that could be produced by $\mathcal{E}_j$ acting on some initial state $\omega'_{\widehat{R}_{j-1}L_jE_{j-1}\CS_1^{j-1}\CP_1^{j-1}\widetilde{E}}$, and $\Sigma_j^\LEAKchann$ is the set of all states that could be produced by $\LEAKchann_j$ for some initial state $\omega''_{R_{j-1}\stabR}$.
		Furthermore, by the bound $H_{2-\frac{1}{\alpha}} \geq H_{\alpha}^\uparrow$~\cite[Corollary~4]{TBH14}, this immediately also gives\footnote{We perform this conversion because the~\cite[Theorem~4.1]{arx_AHT24} proof works mainly with the $H_{\alpha}$ entropies rather than the $H_{\alpha}^\uparrow$ entropies, due to some technicalities in the QES-entropy definition; see~\cite[Remark~4.1]{arx_AHT24}. It may potentially be avoidable by defining a suitable $H_{\alpha}^\uparrow$ version of QES-entropies as discussed in that remark (equivalently, a slight generalization of the concept of \term{$f$-weighted {\Renyi} entropies}), but we leave this for future work.}
		\begin{align}
		H_{2-\frac{1}{\alpha}}(D_1^nS_1^n\copyCS_1^n|\CP_1^nE_n)_\rho \geq \sum_j \left(\inf_{\nu'\in\Sigma_j^\mathcal{E}}H_{\widehat{\alpha}}(D_jS_j\copyCS_j|\CP_1^jE_j\widetilde{E})_{\nu'}-\sup_{\nu'\in\Sigma_j^\LEAKchann}\IDA_\alpha(\stabR ;L_j)_{\nu'}\right).
		\end{align}
		The claim then follows by applying the remaining arguments in the~\cite[Theorem~4.1]{arx_AHT24} proof to relate $			H_{2-\frac{1}{\alpha}}(D_1^nS_1^n\copyCS_1^n|\CP_1^nE_n)_\rho$ and $H_{\widehat{\alpha}}(D_jS_j\copyCS_j|\CP_1^jE_j\widetilde{E})_{\nu'}$ to $H^{f_\mathrm{full}}_{2-\frac{1}{\alpha}}(S_1^n \CS_1^n | \CP_1^n E_n)_\rho$ and  $H^{f_{|\cS_1^{j-1} \cP_1^{j-1}}}_{\widehat{\alpha}}(S_j \CS_j | \CP_j E_j \widetilde{E})_{\nu}$ respectively.
	\end{proof}
\end{theorem}

We note that since the term corresponding to the leakage is independent of any choice of QES in that theorem, all the subsequent simplifications and variations of the above theorem that were discussed in~\cite{arx_AHT24} are still valid if we use GEATTL channels instead, after subtracting $\sum_j \xi_j$ to compensate for the leakage. In particular, we highlight the following core result that one would obtain, as a modification of~\cite[Theorem~5.1]{arx_AHT24}. (In the following bound, a small technicality is that in the $h_{\widehat{\alpha}}$ term we take the minimum\footnote{For readers already familiar with~\cite{arx_AHT24}, we remark that for ease of presentation here we have simply expressed the $h_{\widehat{\alpha}}$ term with a minimization over $j$ rather than introducing the concept of rate-bounding channels. As discussed below~\cite[Definition~5.1]{arx_AHT24}, these perspectives are essentially equivalent, since \emph{one} of the channels in the sequence always yields a rate-bounding channel.} over $j$ (after which it is multiplied by $n$), whereas in contrast, the $\xi_j$ terms are simply summed over. This is because our preceding analysis has already ``extracted'' the $\sum_j \xi_j$ correction to account for leakage, and this is then independent of all subsequent aspects of the analysis, in particular the need to minimize over $j$ when computing the $h_{\widehat{\alpha}}$ term in the~\cite{arx_AHT24} analysis.)

\begin{theorem}\label{th:GREAT}
Let $\rho$ be a state generated by a sequence of GEATTL channels $\{\EATchann_j\}_{j=1}^n$ (Definition~\ref{def:GEATTL}), so they have the form $\EATchann_j=\mathcal{N}_j\circ\LEAKchann_j$. Suppose all the registers $\CS_j$ (resp.~$\CP_j$) are isomorphic to a single register $\CS$ (resp.~$\CP$) with alphabet $\alphCS$ (resp.~$\alphCP$).
Take any $\alpha \in (1,2)$
and let $\widehat{\alpha}=1/(2-\alpha)$.
Suppose furthermore that $\rho = p_\Omega \rho_{|\Omega} + (1-p_\Omega) \rho_{|\overline{\Omega}}$ for some $p_\Omega \in (0,1]$ and normalized states $\rho_{|\Omega},\rho_{|\overline{\Omega}}$. 
Let $S_\Omega$ be a convex set of probability distributions on the alphabet $\alphCS \times \alphCP$, such that for all $\cS_1^n \cP_1^n$ with nonzero probability in $\rho_{|\Omega}$, the frequency distribution $\freq_{\cS_1^n \cP_1^n}$ lies in $S_\Omega$.
Then letting $\bsym{\sigma}_{\CS\CP}$ denote the distribution on $\CS\CP$ induced by any state $\sigma_{\CS\CP}$, we have 
\begin{align}\label{eq:GREAT}
\begin{gathered}
H^\uparrow_\alpha(S_1^n \CS_1^n | \CP_1^n E_n)_{\rho_{|\Omega}} \geq n h_{\widehat{\alpha}} - \left(\sum_j \xi_j\right)
- \frac{\alpha}{\alpha-1} \log\frac{1}{p_\Omega},\\
\text{where}\quad h_{\widehat{\alpha}} = \min_j \inf_{\mbf{q} \in S_\Omega} \inf_{\nu\in\Sigma_j^\mathcal{N}} \left( \frac{1}{\widehat{\alpha}-1}D\left(\mbf{q} \middle\Vert \bsym{\nu}_{\CS\CP}\right)-\sum_{\cS\cP\in \supp(\bsym{\nu}_{\CS\CP})}q(\cS\cP)D_{\widehat{\alpha}}\left(\nu_{S_j E_j \widetilde{E} \land \cS\cP} \middle\Vert \id_{S_j} \otimes\nu_{E_j\widetilde{E} \land \cP} \right) \right),
\end{gathered}
\end{align}
where $\xi_j$ is as defined in~\eqref{eq:chainQES}, and $\Sigma_j^\mathcal{N}$ denotes the set of all states of the form $\mathcal{N}_j\left[\omega_{\widehat{R}_{j-1} L_j E_{j-1} \widetilde{E}}\right]$ for some initial state $\omega \in \dop{=}(\widehat{R}_{j-1} L_j E_{j-1} \widetilde{E})$, with $\widetilde{E}$ being a register of large enough dimension to serve as a purifying register for any of the $\widehat{R}_{j-1} L_j E_{j-1} $ registers.
\end{theorem}

\begin{proof}
As mentioned above, the proof of this result follows by simply applying the same chain of arguments used to prove~\cite[Theorem~5.1]{arx_AHT24}, as none of them further relied on the NS conditions. A small technical difference is that rather than starting from the final Theorem~\ref{thrm:GREATleakage} statement here, in this case we would instead start from the intermediate bound~\eqref{eq:GEATbound_notest} in the above proof, to avoid the change of {\Renyi} parameter (the proof of~\cite[Theorem~5.1]{arx_AHT24} can work directly with $H_\alpha^\uparrow(D_1^nS_1^n\copyCS_1^n|\CP_1^nE_n)_\rho$ rather than $H_\alpha(D_1^nS_1^n\copyCS_1^n|\CP_1^nE_n)_\rho$, unlike the proof of~\cite[Theorem~4.1]{arx_AHT24}).
\end{proof}

The various lower bounds on $h_{\widehat{\alpha}}$ described in~\cite[Lemmas~5.1--5.2]{arx_AHT24} would still be valid as well, along with the convexity properties of the optimization in the $h_{\widehat{\alpha}}$ formula after introducing ``purifying functions''.

\printbibliography

\end{document}